\documentclass[pra,aps,showpacs,onecolumn,twoside,superscriptaddress]{revtex4}

\usepackage{amsmath,amsfonts,amssymb,color,epsfig,graphics,graphicx,latexsym,revsymb,theorem,verbatim}

\newtheorem{definition}{Definition}
\newtheorem{proposition}[definition]{Proposition}
\newtheorem{lemma}[definition]{Lemma}

\newtheorem{theorem}[definition]{Theorem}
\newtheorem{corollary}[definition]{Corollary}
\newtheorem{conjecture}[definition]{Conjecture}

\newtheorem{remark}[definition]{Remark}
\newtheorem{example}[definition]{Example}

\def\squareforqed{\hbox{\rlap{$\sqcap$}$\sqcup$}}
\def\qed{\ifmmode\squareforqed\else{\unskip\nobreak\hfil
\penalty50\hskip1em\null\nobreak\hfil\squareforqed
\parfillskip=0pt\finalhyphendemerits=0\endgraf}\fi}
\def\endenv{\ifmmode\;\else{\unskip\nobreak\hfil
\penalty50\hskip1em\null\nobreak\hfil\;
\parfillskip=0pt\finalhyphendemerits=0\endgraf}\fi}
\newenvironment{proof}{\noindent \textbf{{Proof.~} }}{\qed}

\def\bcj{\begin{conjecture}}
\def\ecj{\end{conjecture}}
\def\bcr{\begin{corollary}}
\def\ecr{\end{corollary}}
\def\bd{\begin{definition}}
\def\ed{\end{definition}}
\def\bea{\begin{eqnarray}}
\def\eea{\end{eqnarray}}
\def\bem{\begin{enumerate}}
\def\eem{\end{enumerate}}
\def\bim{\begin{itemize}}
\def\eim{\end{itemize}}
\def\bl{\begin{lemma}}
\def\el{\end{lemma}}
\def\bpf{\begin{proof}}
\def\epf{\end{proof}}
\def\bpp{\begin{proposition}}
\def\epp{\end{proposition}}
\def\br{\begin{remark}}
\def\er{\end{remark}}
\def\bt{\begin{theorem}}
\def\et{\end{theorem}}

\def\a{\alpha}
\def\b{\beta}

\def\k{\kappa}

\def\p{\pi}
\def\r{\rho}
\def\s{\sigma}

\def\ph{\varphi}

\def\ps{\psi}

\def\G{\Gamma}

\newcommand{\nc}{\newcommand}

\nc{\cA}{{\cal A}} \nc{\cB}{{\cal B}} \nc{\cC}{{\cal C}}
\nc{\cD}{{\cal D}} \nc{\cE}{{\cal E}} \nc{\cF}{{\cal F}}
\nc{\cG}{{\cal G}} \nc{\cH}{{\cal H}} \nc{\cI}{{\cal I}}
\nc{\cJ}{{\cal J}} \nc{\cK}{{\cal K}} \nc{\cL}{{\cal L}}
\nc{\cM}{{\cal M}} \nc{\cN}{{\cal N}} \nc{\cO}{{\cal O}}
\nc{\cP}{{\cal P}} \nc{\cR}{{\cal R}} \nc{\cS}{{\cal S}}
\nc{\cT}{{\cal T}} \nc{\cX}{{\cal X}} \nc{\cZ}{{\cal Z}}

\def\dim{\mathop{\rm Dim}}

\def\max{\mathop{\rm max}}

\def\rank{\mathop{\rm rank}}

\def\tr{\mathop{\rm Tr}}

\def\dg{\dagger}

\def\ox{\otimes}

\newcommand{\bra}[1]{\langle#1|}
\newcommand{\ket}[1]{|#1\rangle}
\newcommand{\proj}[1]{| #1\rangle\!\langle #1 |}
\newcommand{\ketbra}[2]{|#1\rangle\!\langle#2|}
\newcommand{\braket}[2]{\langle#1|#2\rangle}

\newcommand{\jmp}{J. Math. Phys.}

\begin{document}
\title{Distillability and PPT
entanglement of low-rank quantum states}

\author{Lin Chen}
\affiliation{Centre for Quantum Technologies, National University of
Singapore, 3 Science Drive 2, Singapore 117542}
\email{cqtcl@nus.edu.sg (Corresponding~Author)}

\def\Dbar{\leavevmode\lower.6ex\hbox to 0pt
{\hskip-.23ex\accent"16\hss}D}
\author {{ Dragomir {\v{Z} \Dbar}okovi{\'c}}}

\affiliation{Department of Pure Mathematics and
Institute for Quantum Computing,
University of Waterloo, Waterloo, Ontario, N2L 3G1, Canada}
\email{djokovic@uwaterloo.ca}

\begin{abstract}
The bipartite quantum states $\r$, with rank strictly smaller
than the maximum of the ranks of the reduced states $\r_A$
and $\r_B$, are distillable by local operations and classical
communication \cite{hst03}.
Our first main result is that this is also true for NPT states with
rank equal to this maximum. (A state is PPT if the partial transpose
of its density matrix is positive semidefinite,
and otherwise it is NPT.)
This was conjectured first in 1999 in the special case when the ranks
of $\r_A$ and $\r_B$ are equal (see \cite[arXiv preprint]{hst03}).
Our second main result provides a complete solution of the
separability problem for bipartite states of rank 4. Namely,
we show that such a state is separable if and only if it is PPT
and its range contains at least one product state.
We also prove that the so called checkerboard states are distillable
if and only if they are NPT.
\end{abstract}

\date{ \today }




\maketitle



\section{\label{sec:introduction} introduction}

Pure quantum entanglement plays the essential role in various
quantum information tasks, such as GHZ states in quantum
teleportation \cite{bbc93} and graph states in quantum computation
\cite{br01}. However pure entangled states are always coupled with
the environment due to the unavoidable decoherence. As a result they
become mixed states which usually cannot be directly used, or become
quite useless for quantum information tasks \cite{sf11}.

Hence, extracting pure entanglement from mixed states is a basic
task in quantum information. Formally, this task is called {\em
entanglement distillation} and the entanglement measure quantifying
the asymptotically obtainable pure entanglement is called {\em
distillable entanglement}; the states from which we can obtain pure
entanglement are {\em distillable} \cite{bds96}. Apart from the
mentioned physical applications, the distillable entanglement is the
lower bound of many well-known entanglement measures such as the
entanglement cost, entanglement of formation \cite{bds96} and the
squashed entanglement \cite{cw04,bcy10}. An upper bound for
distillable entanglement is provided by the distillable key
\cite{hhh05,csw10}, and recently the gap between them has been
experimentally observed \cite{dkd11}. A lot of effort has been
devoted to the investigation of entanglement distillation, for a
review see \cite{hhh09}.

For a state $\r$ acting on $\cH_A \ox \cH_B$ the partial transpose,
computed in an orthonormal (o.n.) basis $\{\ket{a_i}\}$ of system A,
is defined by
$\r^\G :=\sum_{ij} \ketbra{a_i}{a_j} \ox \bra{a_j}\r\ket{a_i}$.
Mathematically, we say that $\r$ is {\em
1-distillable} if there exists a pure bipartite state $\ket{\ps}$ of
Schmidt rank 2 such that $\bra{\ps} \r^{\G} \ket{\ps} < 0$
\cite{dss00}. More generally, we say that $\r$ is {\em
$n$-distillable} if the state $\r^{\ox n}$ is 1-distillable.
Finally, we say that $\r$ is {\em distillable} under local
operations and classical communication (LOCC) if it is
$n$-distillable for some $n\ge1$.

It is a famous open question whether all bipartite NPT states, i.e.,
the states $\r$ such that $\r^\G$ has at least one negative
eigenvalue, are distillable. Although many papers deal with the
problem of distillation of pure entanglement from NPT states
~\cite{clarisse05,pph07,br03,vd06,hhh09}, the complete solution is
still unknown. In this paper we will solve an open question, which
is an important special case of the distillation problem.

It follows easily from the definition of distillability given above
that bipartite PPT states, i.e., states with positive semidefinite
partial transpose, are not distillable by LOCC
\cite{horodecki97,hhh98prl}. In recent years, PPT
entangled states have been extensively studied in connection with
the phenomena of entanglement activation and universal usefulness
\cite{hhh99,masanes06}, the distillable key \cite{hhh05}, the
symmetry permutations \cite{tg09} and entanglement witnesses
\cite{lkc00}, both in theory and experiment \cite{ab09}.

Therefore, it is an important and basic question to decide whether a
given PPT state is separable \cite{hhh09}. The famous partial
transpose criterion says that the separable states are PPT
\cite{peres96}. Furthermore Horodeckis showed that this is necessary
and sufficient for states in $M\ox N$ quantum systems with
$MN\le6$ \cite{hhh96}.
However the problem becomes difficult for any bigger dimension. For
instance, the first examples of PPT entangled states acting on
$2\ox 4$ and $3\ox 3$ were discovered in 1997 \cite{horodecki97},
but it is still an open problem to decide which PPT states in these systems are separable. The situation is not much better if we
consider the separability problem for bipartite states of fixed rank,
where the cases of ranks 4 and 5 remained unresolved. We have
discovered a simple criterion for separability of states of rank 4.

Throughout the paper the spaces $\cH_A$ and $\cH_B$ are
finite-dimensional. We will ignore the normalization condition
for states since it does not affect the process of entanglement
distillation. Thus, unless stated otherwise, we assume that
the states are non-normalized, and we set $M=\dim \cH_A$
and $N=\dim \cH_B$. As our definition of an ``$M\times N$ state''
can be easily misinterpreted, we state it formally.

\bd
A bipartite state $\r$ is an {\em $M\times N$ state} if
the reduced states $\r_A=\tr_B(\r)$ and $\r_B=\tr_A(\r)$ have
ranks $M$ and $N$, respectively.
\ed

Our first main result, Theorem \ref{thm:distillation=MxNrankN},
asserts that $M\times N$ rank-$N$ NPT states are distillable under
LOCC. The stronger statement that these states are in fact
1-distillable under LOCC is an immediate consequence of the proof.
The special case $M=N$ was proposed in 1999 as a conjecture
\cite[Conjecture 1]{hst03}. As pointed out in \cite{hst03}, it
follows from this special case that all rank 3 entangled states are
distillable. Hence we recover the main result of \cite{cc08}. We
give the details in Sec. II and III.

Our second main result, Theorem \ref{thm:PPTrank4}, solves the
separability problem for bipartite states of rank 4. Namely, it
asserts that a bipartite state of rank 4 is separable if and only if
it is PPT and its range contains at least one product state. This is
a numerically operational criterion, e.g., see the method introduced
in \cite{hlv00}. Therefore, $3\times 3$ PPT states of rank 4
constitute the first class of PPT states of fixed dimension and rank
for which the separability problem can be operationally decided.
We also discuss the analytical method based on the Pl{\"u}cker
coordinates and the Grassmann variety. We give the details in Sec.
IV, V, and the appendix.

We also apply our results to characterize the quantum correlations
inside a tripartite pure state $\ket{\ps}$.
In Sec. \ref{sec:Applicat}, Theorem\ref{thm:nonnonX=SSX}, we show
that the reduced density operators $\r_{AB}$ and $\r_{AC}$ of
$\r=\proj{\ps}$ are undistillable if and only if they are separable.
In other words, there is no tripartite pure state with two undistillable entangled reduced density operators.
We also discuss the relation between distillability and
quantum discord \cite{ccm10}, where the latter is another kind of
quantum correlation.

To investigate the states beyond $M\times N$ states of rank $N$, we
study the family of checkerboard states. This family consists of
two-qutrit states \cite{dzd11}, generically of rank 4, which
generalizes the family constructed by Bru{\ss} and Peres \cite{bp00}.
In Sec. \ref{sec:distilling3x3rank4} we prove that the NPT
checkerboard states are 1-distillable. In Sec. \ref{sec:full-rank}
we analyze further the full-rank properties (defined below) used
in the proof of our first main result. We conclude and discuss our
results in Sec. \ref{sec:conclusion}.

For convenience and reference, we list some known results on
distillability of bipartite NPT states under LOCC:
 \bem
\item[(A)]
 The $2\times N$ NPT states are distillable
\cite{bbp96PRL,hhh97prl,dss00}.
\item[(B)]
  The states violating the reduction criterion are distillable \cite{hh99}.
\item[(C)]
  The $M\times N$ states of rank less than $M$ or $N$ are
distillable \cite{hst03}.
\item[(D)]
 The NPT states of rank at most three are distillable
\cite{cc08}.
\item[(E)]
 All NPT states can be converted by LOCC into NPT Werner states  \cite{werner89,dcl00,dss00}.
 \eem

The states mentioned in (A-D) are in fact 1-distillable.
For the states in result (A), it suffices to observe that generic
pure states in a $2\ox N$ space have Schmidt rank 2.
For (B) see \cite{clarisse05} or the next section. The
1-distillability in result (C) follows from that of (B) and the
proof of \cite[Theorem 1]{hst03}. The assertion for states in result
(D) can be deduced from those of (A) and (B) \cite{cc08}.

In view of the result (E), one might hope to solve the problem of
entanglement distillation by distilling the NPT Werner states.
Unfortunately, although this problem has been studied extensively in
the past decade, it still remains an open and apparently very hard
problem \cite{pph07}. It is thus of high importance to address the
non-Werner states and propose useful tools for their distillation.
In particular, one can see that the result (B), namely reduction
criterion, plays quite important role for distillation of generic
entangled states. For example, we have argued that other three main
results (A), (C), (D) are partially or totally derivable from the
reduction criterion.

Let us introduce the right full-rank property via the Hermitian
observable
$\tr_A(\proj{x}\r)=\bra{x}\r\ket{x}$ where $\ket{x}\in\cH_A$.
We say that $\r$ has the {\em right full-rank property (RFRP)} if
the operator $\bra{x}\r\ket{x}$ is invertible (i.e., has full rank) for some $\ket{x}\in\cH_A$. One defines similarly the {\em left
full-rank property (LFRP)} by using the Hermitian observable
$\tr_B(\proj{y}\r)=\bra{y}\r\ket{y}$ with $\ket{y}\in\cH_B$. Every
result about RFRP has its analog for LFRP, and we shall work mostly
with the former.

We point out that RFRP has appeared previously in
\cite{hlv00} where it was shown that all $M\times N$ $(M\le N)$
PPT states of rank $N$ have RFRP (see the proof of Theorem 1
in that paper).

We prove in Theorem~\ref{thm:distillation=nonFRP} that bipartite
states violating LFRP or RFRP are distillable under LOCC.
We refer to this result as the {\em full-rank criterion}.
It is a crucial ingredient in the proof of
Theorem \ref{thm:distillation=MxNrankN}.
Let us also mention that Theorems \ref{thm:distillation=nonFRP}
and \ref{thm:distillation=MxNrankN} remain valid if we replace
the word ``distillable'' with ``1-distillable'' (see the next
section).

The questions of separability and distillability for the $M\times N$
states $\r$ of rank $R\le N$, with $M\le N$, have very simple answer:
they are separable if and only if they are PPT, and they are
distillable if and only if they are NPT.
The result on separability is proved in \cite{hlv00}, and the
one on distillability in \cite[Theorem 1]{hst03} when
$R<N$, while our first main result handles the case $R=N$.

As a transition between the first and second main result, we
introduce the concepts of reducible and irreducible quantum states.
The sum $\r$ of bipartite states $\r_i$ is {\em B-direct} if
$\cR(\r_B)$ is a direct sum of $\cR((\r_i)_B)$. A bipartite state
$\r$ is {\em reducible} if it is a B-direct sum of two states and
otherwise $\r$ is {\em irreducible}. In the language of quantum
state transformation, a reducible state is stochastic-LOCC (SLOCC)
equivalent to the sum of $\r_i$ whose reduced density operators
$(\r_i)_B$ are pairwise orthogonal \cite{dvc00}. We formally state
this fact in Proposition \ref{pp:def-red}. More properties of
reducible states are also introduced, e.g., to study the
distillability in Proposition \ref{pp:reducible=SUMirreducible},
\ref{pp:irreducibledistillable,commonkernel} and PPT, separability
in Corollary \ref{cor:reducible=SUMirreducible}.

Based on the reducibility, we can easily decide the separability and
distillability of the bipartite state $\r$ of rank 4 when there is
some $\ket{x}\in\cH_A$ such that the operator $\bra{x}\r\ket{x}$ is
non-invertible. This is demonstrated in Lemma
\ref{le:3x3rank4,rankC=1}. We solve the same problem when there is a
product state in the range of $\r$ in Proposition
\ref{pp:3x3rank4,productstate}. In terms of these results we reach
our second main result. That is, a bipartite state of rank 4 is
separable if and only if it is PPT and there is a product state in
its range. This is illustrated in Theorem \ref{thm:PPTrank4} and it
is numerically operational. We also propose an analytical method in
terms of the Grassmann variety to implement our result. In the case
of 2-dimensional subspaces of a $2\ox 3$ space we give a simple
equation which is satisfied if and only if the subspace contains a
product state. Similarly, in the appendix we give a single equation
test for the existence of a product state in the 3-dimensional
subspace of a $2\ox 4$ space. Such equation also exists in the case
of 4-dimensional subspaces of a $3\ox 3$ space but we were not able
to compute it. The method is totally analytical and does not rely on
numerical estimation. The reducibility also constitutes the
technical bases for the distillability of checkerboard states, as
proved in Theorem \ref{thm:checkboard}.

We also apply the main results to derive a few physical facts. For
example, we prove in Theorem \ref{thm:nonnonX=SSX} a tripartite pure
state cannot have two undistillable entangled reduced density
operators. This is realized by the fact that the two reduced density
operators have a maximal local rank equal to its rank, as well as
the first main result and Proposition \ref{pp:PPTrankN}.

Throughout the paper we write $I_k$ resp. $0_k$ for the
identity resp. zero $k\times k$ matrix. The inequality $H\ge0$
means that $H$ is a positive semidefinite Hermitian operator
or matrix. Similarly, $H>0$ means that $H$ is positive definite.
We denote by $\cR(\r)$ the range space of an operator $\r$.

We write $[X,Y]$ for the commutator $XY-YX$. We recall that
if $X$ is normal and $Y$ arbitrary then $[X,Y]=0$ implies
$[X^\dag,Y]=0$. This follows from the fact that, for a normal
operator $X$, there exists a polynomial $f(t)$ such that
$X^\dag=f(X)$.

\section{\label{sec:preliminary} preliminary facts}

Before proceeding with the proofs of our results, let us recall
some additional facts. First, by using an o.n. basis
$\{\ket{u_i}\}$ of $\cH_A$, we can write any state as
$\r = \sum^M_{i,j=1} \ketbra{u_i}{u_j} \ox \s_{ij}$.
Hence $\r$ is represented by its matrix
\begin{equation}
\label{eq:mat-rho}
\r=\left(
    \begin{array}{ccccc}
      \s_{11} & \s_{12}  & \cdots & \s_{1M} \\
      \s_{21} & \s_{22}  &        & \s_{2M} \\
      \vdots  &          &        &         \\
      \s_{M1} & \s_{M2}  &        & \s_{MM} \\
    \end{array}
  \right).
\end{equation}
Usually we assume that $\{\ket{u_i}\}$ is the computational
basis $\{\ket{i}\}$.

Second, the RFRP (defined in the Introduction) was an essential tool
used in \cite{hlv00} and will also play an important role in this
paper. Clearly, RFRP is equivalent to the following assertion: One
can choose an o.n. basis $\{\ket{u_i}\}$ of $\cH_A$ such that some
diagonal block $\s_{ii}$ in the above matrix has rank $N$.

The condition that $\rank(\bra{x}\r\ket{x})=N$ is equivalent to
$\ker \bra{x}\r\ket{x}=0$. It follows that $\r$ violates RFRP if and
only if for each state $\ket{x}\in\cH_A$ there exists a state
$\ket{y}\in\cH_B$ such that $\ket{x}\ox\ket{y}\in\ker\r$.

\begin{example} \label{sep->FRP}
{\rm We use an argument from the proof of \cite[Lemma 3]{hlv00},
to show that all $M\times N$ separable states $\r$ have RFRP. We
can write $\r$ as a finite sum of non-normalized product states
$\r=\sum_i \proj{a_i,b_i}$. Choose $\ket{x}\in\cH_A$ such that
$\braket{x}{a_i}\ne0$ for all $i$. As the $\ket{b_i}$ span $\cH_B$,
$\bra{x}\r\ket{x}$ has rank $N$.

We will show in the next section that this result extends to all
PPT states. }
\end{example}

We shall need the following result proved in \cite{hlv00}.
\bpp
\label{pp:PPTrankN}
If $\r$ is an $M\times N$ PPT state with $M\le N$ and rank $N$,
then $\r$ is separable and can be written as a sum of $N$
product states
$$ \r=\sum^N_{i=1}\proj{a_i,b_i}. $$
\epp

If $A$ and $B$ are linear operators on $\cH_A$ and $\cH_B$,
respectively, then we shall refer to $A\ox B$ as a {\em local
operator}. The abbreviation ILO will refer to invertible local
operators, i.e., to operators $A\ox B$ with both $A$ and $B$
invertible. It is easy to see that RFRP is invariant under ILOs,
i.e., if $A\ox B$ is an ILO and the state $\r$ has RFRP then the
transformed state $(A\ox B)^\dag ~\r~ (A\ox B)$ also has RFRP.
Moreover, if $\r$ violates RFRP then so does $(A\ox B)^\dag ~\r~
(A\ox B)$ even if $A\ox B$ is not invertible.

Third, there is a simple way to prove distillability of $\r$
which can be applied in many cases. For that purpose observe
that if $\r$ is given by its matrix (\ref{eq:mat-rho}) then
$$ \r^\G=\left(
    \begin{array}{ccccc}
      \s_{11} & \s_{21}  & \cdots & \s_{M1} \\
      \s_{12} & \s_{22}  &        & \s_{M2} \\
      \vdots  &          &        &         \\
      \s_{1M} & \s_{2M}  &        & \s_{MM} \\
    \end{array}
  \right). $$
Since $\r$ is a Hermitian matrix, so is $\r^\G$.
Let $X=\left(\begin{array}{cc}a&b\\b^*&c\end{array}\right)$ be
a principal $2\times2$ submatrix of $\r^\G$. Thus $a$ and $c$
are diagonal entries of $\r^\G$.
\bl
\label{le:triv-dist}
If $ac=0$ while $b\ne0$ then $\r$ is distillable.
\el
\bpf
Since $X$ has a negative eigenvalue and the diagonal blocks
$\s_{ii}\ge0$, the diagonal entries $a$ and $c$ must belong to
different diagonal blocks, say $\s_{kk}$ and $\s_{ll}$. Let
$P$ be the orthogonal projector onto the 2-dimensional
subspace of $\cH_A$ spanned by $\ket{k}$ and $\ket{l}$.
Then the projected state $\r':=(P\ox I_N)~\r~(P\ox I_N)$ is
an NPT state acting on a $2\ox N$ system. Hence $\r'$ is
distillable by the result (A). Consequently, $\r$ is also
distillable.
\epf

Whenever the distillability of $\r$ can be proved by direct
application of this simple lemma, we shall say that the state $\r$
is {\em trivially distillable}. For example, by using this concept
we generalize Theorem 2 of Ref. \cite{hst03}.
 \bt
 \label{thm:onlyonepureentangled}
Let $\r=\proj{\ps}+\s$ be an $M\times N$ state where $\s$ is
any state and $\ket{\ps}$ is a pure entangled state.
If $r:=\rank(\s_A)<M$, then $\r$ is distillable.
 \et
 \bpf
We may assume that $\{\ket{1},\ldots,\ket{r}\}$ is a basis of
$\cR(\s_A)$. We can write $\ket{\psi}=\sum_{i=1}^M \ket{i,\psi_i}$
where $\ket{\psi_i}\in\cH_B$. Since $r<M$, we must have
$\ket{\ps_M}\ne0$. Since $\ket{\psi}$ is entangled, at least one of
the $\ket{\psi_k}$, $k<M$, is not parallel to $\ket{\ps_M}$. Let
us fix such index $k$. Let $V$ be an ILO on $\cH_B$ such that
$V\ket{\psi_k}=\ket{1}$ and $V\ket{\psi_M}=\ket{2}$.
If $P=\proj{k}+\proj{M}$, then the state
$\r':=(P\ox V)~\r~(P\ox V)^\dag$ acts on a $2\ox N$ system.

We claim that $\r'$ is trivially distillable. Since this is the
first time that we are applying Lemma \ref{le:triv-dist}, we shall
give all the details.
Since $P\ox V\ket{\psi}=\ket{k,1}+\ket{M,2}$, we have
\bea \notag P\ox V~\proj{\psi}~P\ox V^\dag &=&
\proj{k,1}+\proj{M,2}+\ketbra{k,1}{M,2}+\ketbra{M,2}{k,1}, \\
\notag \left( P\ox V~\proj{\psi}~P\ox V^\dag \right)^\G &=&
\proj{k,1}+\proj{M,2}+\ketbra{M,1}{k,2}+\ketbra{k,2}{M,1}. \eea
On the other hand since $\s~(\proj{M}\ox V)^\dag=0$, the
$2N\times2N$ matrix of $(P\ox V)~\s~(P\ox V)^\dag$ has its
nonzero entries all in the $N\times N$ submatrix contained in the
first $N$ rows and columns. Hence, the $2\times2$ principal
submatrix of $(\r')^\G$ in rows 2 and $N+1$ has the form
$\left(\begin{array}{cc}*&1\\1&0\end{array}\right)$,
which proves our claim. It follows from this claim that
$\r$ is also distillable.
 \epf

Theorem 2 of Ref. \cite{hst03} is obtained by taking $\s$ to be
a pure product state.

In the following proposition we represent $M\times N$ states in
a convenient matrix form which will be used in our proofs.
We include a result proved in \cite{hlv00} for the states
which are also PPT.

\begin{proposition} \label{pp:pocetak}
Let $\r$ be an $M \times N$ state of rank $R$.

(i) $\r$ can be written as
\bea \label{ro}
\r &=& \sum_{i,j} \ketbra{i}{j} \ox C_i^\dag C_j
=(C_1,\ldots,C_M )^\dg \cdot (C_1,\ldots,C_M),
\eea
where $C_i$ are $R\times N$ matrices such that
$\r_B=\sum_i C_i^\dag C_i>0$.

 (ii) If $R=N$ and $\r$ has RFRP, then there exists an invertible
local operator $A\ox B$ such that
\begin{eqnarray*}
\r' &:=& (A\ox B)^\dg ~\r~ A\ox B \\
&=& (C_1,\ldots,C_{M-1},I_N)^\dg\cdot(C_1,\ldots,C_{M-1},I_N).
\end{eqnarray*}
Moreover, $\r'$ is PPT if and only if the $C_i$ are pairwise
commuting normal matrices.
\end{proposition}

\bpf The assertion (i) follows from the spectral decomposition
theorem. Indeed, by that theorem we have
\begin{equation} \label{jed:ro-psi}
\r=\sum^R_{i=1} \proj{\psi_i}, \quad
\ket{\psi_i}=\sum^M_{j=1} \ket{j,\psi_{ij}},
\end{equation}
where $\ket{\psi_i}$ are non-normalized pure states. Since these
states span the range of $\r$, they must be linearly independent.
However, they are not uniquely determined by $\r$.
The matrices $C_j$ can be chosen as follows:
\begin{equation} \label{jed:mat-Cj}
C_j=\left( \ket{\psi_{1j}},\ldots,\ket{\psi_{Rj}} \right)^\dag,
\quad j=1,\ldots,M.
\end{equation}
In other words, $\bra{\psi_{ij}}$ is the $i$th row of $C_j$. Since
$\r_B\ge0$ and $\rank\r_B=M$, we have $\r_B>0$.

For the first assertion of (ii) we may assume that $C_M$ is
invertible and then apply the local operator $I_A \ox C_M^{-1}$.
It is shown in \cite{hlv00} that the PPT condition implies
that the $C_i$s are pairwise commuting normal matrices.
The converse is straightforward.
\epf

The next example shows that there exist $M\times N$ NPT states of
rank $N$ which violate RFRP.
\begin{example}
{\rm
For the (non-normalized) $3\times 3$ rank-3 antisymmetric state
\begin{equation} \label{jed:ro-as}
\r_{as}=\sum_{1\le i<j\le3}
    (\ket{ij}-\ket{ji})(\bra{ij}-\bra{ji}),
\end{equation}
the operator $\bra{x}\r_{as}\ket{x}$ has rank 2 for all
nonzero vectors $\ket{x}_A=\sum \xi_i \ket{i}\in\cH_A$.
This can be verified as follows.
We first compute the matrices $C_i$. In this example we can
choose $\ket{\psi_i}=\ket{jk}-\ket{kj}$ where $(i,j,k)$ is a
cyclic permutation of $(1,2,3)$. Then we find that
\begin{equation} \label{jed:Ci}
C_1=
\left(\begin{array}{ccc}0&0&0\\0&0&-1\\0&1&0\end{array}\right),
\quad C_2=
\left(\begin{array}{ccc}0&0&1\\0&0&0\\-1&0&0\end{array}\right),
\quad C_3=
\left(\begin{array}{ccc}0&-1&0\\1&0&0\\0&0&0\end{array}\right).
\end{equation}
Thus $\bra{x}\r_{as}\ket{x}=X^\dag X$ where $X=\sum_i \xi_i C_i$.
Since $X$ is antisymmetric its rank must be even, and so
both $X$ and $\bra{x}\r_{as}\ket{x}$ have rank 2.
}
\end{example}
For a more general class of examples violating RFRP see
section \ref{sec:full-rank}.
The problem of characterizing the states violating RFRP remains open.

Let us make two important observations. First, there is a
stronger version of the result (B) from the Introduction,
due to Clarisse \cite{clarisse05}.
Namely, he has shown that a bipartite state $\r$ which violates
the reduction criterion is in fact 1-distillable \cite{dss00}.
This is much stronger then being merely distillable.

Second, assume that a bipartite state $\s$ is obtained from $\rho$
by applying a local operator: $\s=(A\ox B)^\dag ~\r~ (A\ox B)$,
where $A\ox B$ may be singular. If $\s$ is 1-distillable
then $\r$ is also 1-distillable. Indeed if the Hermitian
operator $D$ is a 1-distillation witness detecting $\s$, i.e.,
$\tr(\s D)<0$ while $\tr(\s'D)\ge0$ for all 1-undistillable
states $\s'$, then $(A\ox B)~D~(A\ox B)^\dag$ is a 1-distillation
witness detecting $\r$.

One can easily verify that in all cases that arise in our proofs,
the above two observations garantee that the state $\r$ from which
we start is not only distillable but also 1-distillable. Hence
Theorems \ref{thm:distillation=nonFRP} and
\ref{thm:distillation=MxNrankN} remain valid when we
replace the word ``distillable'' with ``1-distillable''.

\section{\label{sec:MxNrankN}
Disitillability of $M\times N$ NPT states of rank $N$}

In this section we prove our main result
Theorem~\ref{thm:distillation=MxNrankN}. However, our first
objective is to show that the states which violate LFRP or RFRP are
distillable. Since PPT states are not distillable,
it follows that all PPT states must possess both LFRP and RFRP.

For convenience, we shall denote by $X[k]$ the submatrix of the
matrix $X$ consisting of the last $k$ columns. If a matrix $X$
factorizes as $X=YZ$, then we shall say that $Z$ is a
{\em right factor} of $X$.

 \bt
\label{thm:distillation=nonFRP}
Bipartite states which violate LFRP or RFRP are distillable.
 \et
 \bpf
It suffices to prove that any $M\times N$ state $\r$ which violates
RFRP is distillable. Let $R$ be the rank of $\r$. If $R<N$, then
$\r$ is distillable by result (C). From now on we assume that $R\ge
N$. By Proposition \ref{pp:pocetak} we have
$\r=(C_1,\ldots,C_M)^\dag \cdot (C_1,\ldots,C_M)$ where each matrix
$C_i$ is of size $ R\times N$, and $\r_B=\sum_i C_i^\dag C_i>0$.

We can replace $(C_1,\ldots,C_M)$ with $U(C_1,\ldots,C_M)$ where $U$ is a unitary matrix, without changing $\r$.
The effect of an invertible local transformation
$\r\to(I\ox B)^\dag~\r~(I\ox B)$ is to replace each $C_i$ by $C_iB$.
Recall that RFRP is preserved by these local transformations.
In order to prove the theorem we can apply these
kind of transformations as many times as needed.

Since $\r_A$ is invertible, each $C_i\ne0$. Let $r_1$ be the
rank of $C_1$. By multiplying $(C_1,\ldots,C_M)$ by a unitary matrix
$U_1$ on the left hand side, we may assume that the last
$R-r_1$ rows of $C_1$ are zero. We choose an invertible matrix
$B_1$ such that $C_1 B_1=I_{r_1}\oplus0$. By multiplying all
$C_i$s by $B_1$ on the right hand side, we may assume that
$C_1=I_{r_1}\oplus0$. Since $\r$ violates RFRP, we have $r_1<N$.

For a sequence of indexes $1\le i_1<i_2<\cdots<i_k\le M$ we
denote by $\r_{i_1,\ldots,i_k}$ the corresponding principal
submatrix of $\r$ of size $kN\times kN$, i.e.,
$$ \r_{i_1,\ldots,i_k}=(C_{i_1},\ldots,C_{i_k})^\dag \cdot
(C_{i_1},\ldots,C_{i_k}). $$
The corresponding principal submatrix of
$\r^\G$ will be denoted by $\r_{i_1,\ldots,i_k}^\G$. For instance,
for $i>1$ we have
$$ \r_{1,i}=( C_1, C_i )^\dag \cdot ( C_1, C_i )=
\left(\begin{array}{cc} C_1^\dag C_1 & C_1^\dag C_i \\ C_i^\dag C_1
& C_i^\dag C_i
\end{array} \right) $$
and
$$ \r_{1,i}^\G = \left(\begin{array}{cc}
C_1^\dag C_1 & C_i^\dag C_1 \\ C_1^\dag C_i & C_i^\dag C_i
\end{array} \right). $$

We split each $C_i$ into four blocks $C_i=\left(\begin{array}{cc}
C_{i1} & C_{i2} \\ C_{i3} & C_{i4} \end{array}\right)$ with $C_{i1}$
square of size $r_1$. Since $C_1=I_{r_1}\oplus0$, we have
$$ \r_{1,i} = \left(\begin{array}{cccc}
I_{r_1} & 0 & C_{i1} & C_{i2} \\
0 & 0 & 0 & 0 \\
C_{i1}^\dag & 0 & * & * \\
C_{i2}^\dag & 0 & * & *
\end{array}\right),\quad i>1, $$
where the asterisk stands for an unspecified block.
If some $C_{i2}\ne0$, then $\r$ is trivially distillable.
Thus we may assume that all $C_{i2}=0$.

Since $\r_B$ has rank $N$, its submatrix $\r_B[N-r_1]$ must have rank
$N-r_1$. Since  $C_{i2}=0$,  $C_{i4}$ is a right factor of the submatrix $(C_i^\dag C_i)[N-r_1]$. Therefore $C_{i4}\ne0$ for at
least one index $i>1$. By permuting the $C_i$s with $i>1$, we may
assume that $C_{24}\ne0$. Let $r_2$ be its rank. By multiplying
$(C_1,\ldots,C_M)$ by a unitary matrix $I_{r_1}\oplus U_2$ on the
left hand side, we may assume that the last $R-r_1-r_2$ rows of
$C_{24}$ are zero. Let $B_2$ be an invertible matrix such that
$C_{24}B_2=I_{r_2}\oplus0$. By multiplying each $C_i$ by
$I_{r_1}\oplus B_2$ on the right hand side, we may
assume that $C_{24}=I_{r_2}\oplus0$. Note that these operations
performed on $(C_1,\ldots,C_M)$ do not alter $C_1$ and that the
equalities $C_{i2}=0$ remain valid.

Assume that $r_1+r_2<N$. We split each $C_{i4}$ into four blocks
$C_{i4}=\left( \begin{array}{cc} C_{i41} & C_{i42} \\ C_{i43} &
C_{i44} \end{array} \right)$, with $C_{i41}$ square of size $r_2$.
For $i>2$ we have
$$  \r_{2,i}=\left(\begin{array}{cc}
C_2^\dag C_2 & C_2^\dag C_i \\ C_i^\dag C_2 & C_i^\dag C_i
\end{array} \right) =
\left(\begin{array}{cccc}
* & C_{23}^\dag C_{24} & * & C_{23}^\dag C_{i4} \\
C_{24}^\dag C_{23} & C_{24}^\dag C_{24} & C_{24}^\dag C_{i3} &
C_{24}^\dag C_{i4} \\
* & C_{i3}^\dag C_{24} & * & C_{i3}^\dag C_{i4} \\
C_{i4}^\dag C_{23} & C_{i4}^\dag C_{24} & C_{i4}^\dag C_{i3} &
C_{i4}^\dag C_{i4}
\end{array}\right). $$
We extract from it the principal submatrix
$$ \left(\begin{array}{cc}
C_{24}^\dag C_{24} & C_{24}^\dag C_{i4} \\
C_{i4}^\dag C_{24} & C_{i4}^\dag C_{i4}
\end{array}\right)=
\left(\begin{array}{cccc}
I_{r_2} & 0 & C_{i41} & C_{i42} \\
0 & 0 & 0 & 0 \\
C_{i41}^\dag & 0 & * & * \\
C_{i42}^\dag & 0 & * & *
\end{array}\right). $$
If some $C_{i42}\ne0$, then $\r$ is trivially distillable.
Thus we may assume that all $C_{i42}=0$.

Since $\r_B$ has rank $N$, the matrix $\r_B[N-r_1-r_2]$ must have
rank $N-r_1-r_2$. Since  $C_{i2}=0$ and $C_{i42}=0$, we have
$$ C_i=\left(\begin{array}{ccc}
C_{i1} & 0 & 0 \\ * & C_{i41} & 0 \\ * & C_{i43} & C_{i44}
\end{array}\right). $$
Consequently, $C_{i44}$ is a right factor of the submatrix
$(C_i^\dag C_i)[N-r_1-r_2]$. Therefore $C_{i44}\ne0$ for at least
one index $i>2$. By permuting the $C_i$s with $i>2$, we may assume
that $C_{344}\ne0$. Let $r_3$ be its rank.
By multiplying $(C_1,C_2,\ldots,C_M)$ by a unitary matrix
$I_{r_1+r_2}\oplus U_3$ on the left hand side , we may assume that
the last $R-r_1-r_2-r_3$ rows of $C_{344}$ are zero. Let $B_3$ be
an invertible matrix such that $C_{344}B_3=I_{r_3}\oplus0$. By
multiplying each $C_i$ by $I_{r_1+r_2}\oplus B_3$ on the right
hand side, we may assume that $C_{344}=I_{r_3}\oplus0$.
Note that these operations performed on $(C_1,C_2,\ldots,C_M)$
alter neither $C_1$ nor $C_{24}$ and that the equalities
$C_{i2}=0$ and $C_{i42}=0$ remain valid.

If $r_1+r_2+r_3<N$ we can continue by splitting each $C_{i44}$
into four blocks, etc.
This process terminates as soon as $\r$ has been shown to be
trivially distillable. Otherwise it can be continued as long as
$r_1+\cdots+r_k<N$ and $k<M$. However, note that if $k$ becomes
equal to $M$ then we must have $r_1+\cdots+r_M=N$ because
$\rank\r_B=N$.

We claim that the process must terminate while the inequality
$r_1+\cdots+r_k<N$ is still valid. Indeed assume that we reach
a point where $r_1+\cdots+r_k=N$. Since $k\le M$ we can set
$\ket{x}=t_1 \ket{1}+\cdots+t_k\ket{k}\in\cH_A$,
where $t_1,\ldots,t_k$ are real parameters. Then we have
$$ \bra{x}\r\ket{x}=\sum_{i,j=1}^k t_i t_j C_i^\dag C_j =
\left(\sum_{i=1}^k t_i C_i\right)^\dag
\left(\sum_{i=1}^k t_i C_i\right). $$
It follows that
$$ \det \bra{x}\r\ket{x} =
\left|\det \left( \sum_{i=1}^k t_i C_i \right) \right|^2. $$
The determinant on the left hand side is a polynomial in the
parameters $t_1,\ldots,t_k$. As $C_{i2}=0$, $C_{i42}=0,\ldots$
for all $i$, the coefficient of $t_1^{2r_1}t_2^{2r_2}\cdots t_k^{2r_k}$
in this determinant is equal to 1. Hence, we can choose the values
of the parameters $t_i$ so that the operator $\bra{x}\r\ket{x}$
becomes invertible. This contradicts the hypothesis that $\r$ violates RFRP, and completes the proof.
 \epf

Let us also record the observation we made earlier.
\bcr
 \label{cor:pptFRP}
All bipartite PPT states possess LFRP and RFRP.
\ecr
For $M\times N$ states with $M\le N$ and rank $N$ this was
proved in \cite{hlv00}.

By using Theorem~\ref{thm:distillation=nonFRP} we can now prove our
main result.
 \bt
 \label{thm:distillation=MxNrankN}
The $M \times N$ NPT states of rank $N$ are distillable under
LOCC.
 \et

 \bpf
By the result (C) we may assume that $M\le N$. In view of
Theorem~\ref{thm:distillation=nonFRP}, it suffices to prove the
assertion for $M\times N$ NPT states $\r$ of rank $N$ which have
RFRP. Since $\r$ is NPT, we must have $M\ge2$. By Proposition
\ref{pp:pocetak} we can assume that $
\r=(C_1,\ldots,C_{M-1},I_N)^\dg\cdot (C_1,\ldots,C_{M-1},I_N)$ where
$C_i$ are $N\times N$ matrices. Let $\r_i=(P_i\ox I)~\r~(P_i\ox I)$
where $P_i=\proj{i}+\proj{M}$, $i<M$. If some $\r_i$ is NPT, then
$\r_i$ (and $\r$) is distillable by the result (A). Thus we may
assume that all $\r_i$ are PPT. By applying Proposition
\ref{pp:pocetak} (ii) to $\r_i$, we deduce that $C_i$ is a normal
matrix. Since $\r$ is NPT, the same proposition implies that there
exist $i,j$ such that $[C_i,C_j]\ne0$. In particular, $M\ge3$. We
may assume that $i=1$ and $j=2$.

Let $\r'=(V\ox I_N)^\dag ~\r~ (V\ox I_N)$ where
$V=(x\ket{1}+\ket{2})\bra{2}+\proj{M}$ and $x$ is a complex
parameter. Obviously we have
$$ (C_1,\ldots,C_{M-1},I_N)(V\ox I_N)=(0,X,0,\ldots,0,I_N), $$
where $X=x C_1+C_2$. Since the range of $\r'_A$ is contained in
the subspace spanned by $\ket{2}_A$ and $\ket{M}_A$, we can view
$\r'$ as a state acting on a $2\ox N$ space. Then its density matrix
is $\r'=(X,I_N)^\dag\cdot(X,I_N)$ and we have
$$ \r'_A=\left(\begin{array}{cc}
\tr (X^\dag X) & \tr (X^\dag) \\ \tr (X) & N \end{array} \right). $$
Hence $\det(\r'_A)=N\tr(X^\dag X)-|\tr(X)|^2$.
Since $[C_1,C_2]\ne0$, the matrices $X$ and $I_N$ are linearly
independent. Consequently, $\det(\r'_A)>0$ by the Cauchy--Schwarz
inequality. Since $\r'_B=I_N+X^\dag X>0$, we conclude that $\r'$ is a $2\times N$ state for all $x$. Evidently its rank is $N$.

Assume that $\r'$ is PPT for all values of $x$. By Proposition
\ref{pp:pocetak}, the matrix $X=x C_1 + C_2$ is normal. Since $x$
is arbitrary and both $C_1$ and $C_2$ are normal, we deduce that
$x[C_1,C_2^\dag]=x^*[C_1^\dag,C_2]$. By setting first $x=1$ and then
that $x$ is the imaginary unit, we deduce that $[C_1,C_2^\dag]=0$.
Consequently, also $[C_1,C_2]=0$ because $C_2$ is normal, and
we have a contradiction. We conclude that $\r'$ must be NPT for
at least one value of $x$. Such $\r'$ is distillable by the
result (A) and, consequently, $\r$ is also distillable.
This completes the proof.
 \epf

This theorem gives a new family of bipartite distillable states
under LOCC. In particular, for $M=N$ we obtain a positive answer to
an open problem proposed in \cite{hst03}.

\section{\label{sec:reducibleirreducible}
Reducible and irreducible states}

In the previous sections we have shown that the states violating
LFRP or RFRP are distillable and discussed several applications of
this result. In this and subsequent sections we would like to explore
further the structure of states having both full-rank properties.
This problem is important for two main reasons.
First, most bipartite quantum states have LFRP and RFRP, i.e., this
is a generic property of the state space (see section
\ref{sec:full-rank}
Second, by Theorem~\ref{thm:distillation=MxNrankN},
the distillability problem for bipartite states reduces to the
case of $M\times N$ states of rank bigger than $\max(M,N)$.
We shall introduce some new concepts such as reducible states,
and use them to construct new classes of distillable states.
They complement Theorems \ref{thm:distillation=nonFRP} and
\ref{thm:distillation=MxNrankN}.

We begin with the definition of reducible and irreducible states.
 \bd
\label{def:red} The sum $\r$ of bipartite states $\r_i$ is {\em
B-direct} if $\cR(\r_B)$ is a direct sum of $\cR((\r_i)_B)$. A
bipartite state $\r$ is {\em reducible} if it is a B-direct sum of
two states and otherwise $\r$ is {\em irreducible}. We denote them
by $\r_{re}$ and  $\r_{ir}$, respectively.
 \ed

It is clear that if $A\ox B$ is an arbitrary ILO, then $\r$ is
reducible if and only if $(A\ox B)~\r~(A\ox B)^\dag$ is reducible.

We observe that in the case $N=3$ the reducible NPT states are
distillable. \bl
 \label{le:Mx3reducibleNPT}
Any $M\times3$ reducible NPT state is distillable. \el
 \bpf
Such a state $\r$ is a B-direct sum of two states with B-local ranks
1 and 2. Since the one with B-local rank 1 is separable, the other
one must be NPT and so it is distillable by result (A).
Consequently, $\r$ is distillable too.
 \epf

The following lemma is obvious.
 \bl
\label{le:reducible=SUMirreducible} Any bipartite state is a finite
B-direct sum of irreducible states.
 \el

We point out that this decomposition into irreducibles is not unique
in general. For instance any $1\times2$ state $\r$ has infinitely
many B-direct decompositions $\r=\r_1+\r_2$, where $\r_1$ and $\r_2$
are product states. Here is another more interesting example with
entangled irreducible summands.
\begin{example}
{\rm Let $\r$ be the $4\times4$ (non-normalized) state which is by
definition the B-direct sum $\r=\r_1+\r_2$ of two irreducible states
$\r_1=2\proj{\phi_1}$ and $\r_2=2\proj{\phi_2}$ where
$\phi_1=\ket{11}+\ket{22}$ and $\phi_2=\ket{13}+\ket{24}$. It admits
another such decomposition $\r=\r'_1+\r'_2$, where
$\r'_1=\proj{\phi'_1}$ and $\r'_2=\proj{\phi'_2}$ with
$$ \phi'_1=\ket{11}+\ket{13}+\ket{22}+\ket{24} \quad \text{and}
\quad \phi'_2=\ket{11}-\ket{13}+\ket{22}-\ket{24}. $$ }
\end{example}

The reducible states have the following important property.

\bpp \label{pp:def-red} Let $\r=\r_1+\r_2$ (a B-direct sum) be a
reducible $M\times N$ state. Then there exists a Hermitian operator
$V>0$ on $\cH_B$ such that the states $\r'_i=(I\ox V)~\r_i~(I\ox V)$
$(i=1,2)$ have orthogonal ranges, i.e., $(\r'_1)_B (\r'_2)_B=0$.
\epp \bpf For convenience, set $\s_i=(\r_i)_B$. Since
$\cH_B=\cR(\s_1)\oplus\cR(\s_2)$, we have $\s:=\s_1+\s_2>0$. We set
$V=\s^{-1/2}$ and $\s'_i=V\s_i V$. Since $\s'_1+\s'_2=I_N$ and
$\rank(I_N-\s'_1)=\rank(\s'_2)=N-\rank (\s'_1)$, it follows that
$\s'_1$ is an orthogonal projector. Hence
$\s'_1\s'_2=\s'_1(I_N-\s'_1)=0$. For $\r'_i=(I\ox V)~\r_i~(I\ox V)$,
we have $(\r'_i)_B=V(\r_i)_B V=V\s_i V=\s'_i$ and so
$(\r'_1)_B(\r'_2)_B=\s'_1\s'_2=0$. \epf

\bcr \label{cor:reducible=SUMirreducible} If $\r=\sum_i\r_i$ is a
B-direct sum, then $\r$ is separable (PPT) if and only if every
$\r_i$ is separable (PPT). In particular, the reducible $3\times3$
state $\r$ is separable if and only if it is PPT.
\ecr

For convenience, we refer to the following states \cite[Sec.
IV]{hhh98acta} as label states
$$ \r_{la}:=\sum_i p_i\proj{a_i,b_i}_{A_1 B_1} \ox
\proj{\psi_i}_{A_2 B_2}, $$ where the product states $\ket{a_i,b_i}$
are distinguishable via LOCC. In particular, when
$\ket{a_i,b_i}=\ket{1,i}$, the state $\r_{la}$ becomes a reducible
state. It is also known that three fundamental entanglement
measures, i.e., the distillable entanglement $E_d$, entanglement
cost $E_c$ and entanglement of formation $E_f$ are all equal for
such states \cite{hhh98acta}. Indeed, because $\ket{a_i,b_i}$ are
distinguishable via LOCC, by measurements Alice and Bob can get each
pure state $\ket{\psi_i}$ with probability $p_i$. Therefore
$E_d(\r_{la})=\sum_i p_i S(\tr_A \proj{\ps_i})$, where $S(\r)$ is
the von Neumann entropy. On the other hand by the definition of
$E_c$ and $E_f$ \cite{bds96}, we have the inequalities
$$E_d(\r_{la}) \le E_c(\r_{la}) \le E_f(\r_{la}) \le \sum_i p_i
S({\tr}_A \proj{\ps_i}). $$ So the three entanglement measures
coincide.

One can prove a similar fact for reducible states $\r=\sum_i
p_i\r_i$, provided that $(\r_i)_B (\r_j)_B = 0$ for $i \ne j$. We
assume here that the state $\r$ and the $\r_i$ are all normalized,
$p_i>0$, $\sum_i p_i=1$, and the sum is B-direct. Thus Alice and Bob
can get each state $\r_i$ with probability $p_i$. We have
 \bea
 \label{ea:threemeasures}
 E_d(\r_{re}) = \sum_i p_i E_d(\r_i) \le E_c(\r_{re}) \le
 E_f(\r_{re}) \le \sum_i p_i E_f(\r_i).
 \eea
Evidently, the three entanglement measures coincide if and only if
$E_d(\r_i)=E_f(\r_i)$ for all $i$, e.g., this is satisfied for the
label states \footnotemark. \footnotetext{However, the equality does
not hold for most of the known mixed bipartite states. The reason is
that, according to the computable results, the distillable
entanglement is smaller than the entanglement cost. It is also an
important open problem to decide whether the equality holds only for
the label states (including the pure states).} Thus we have the
following result.

\bpp  \label{pp:reducible=SUMirreducible}
If $\r=\sum_i\r_i$ is a B-direct sum, then $\r$ is distillable if and only if at least one of the $\r_i$ is distillable.
\epp
 \bpf
By Proposition \ref{pp:def-red}, we may assume that the states
$\r_i$ are distinguishable by LOCC. From the equality in Eq.
(\ref{ea:threemeasures}) we see that the inequality $E_d(\r_{re})>0$
holds if and only if $E_d(\r_i)>0$ for at least one index $i$.
 \epf

Together with Lemma~\ref{le:reducible=SUMirreducible},
this proposition shows that the question of deciding whether
an arbitrary bipartite state is distillable reduces to the case
of irreducible states.
However, the distillability problem for irreducible states is hard.
We tackle a special case in the following proposition.

 \bpp
 \label{pp:irreducibledistillable,commonkernel}
Let $\r$ be an irreducible $M\times N$ state such that
$\cH'_A\ox\ket{b}\subseteq\ker\r$ for some $(M-1)$-dimensional
subspace $\cH'_A\subseteq\cH_A$ and some state $\ket{b}\in\cH_B$.
Then $\r$ is distillable.
 \epp
 \bpf
We may assume that $\cH'_A$ is spanned by the basis vectors
$\ket{i}_A$, $i>1$, and that $\ket{b}=\ket{1}_B$. By Proposition
\ref{pp:pocetak}, we have
$\r=(C_1,\ldots,C_M)^\dg\cdot(C_1,\ldots,C_M)$, where the $C_i$ are
$R\times N$ matrices and $R=\rank\r$. Moreover, the first columns of
the $C_i$ are 0 for $i>1$. Since $\r_B$ is invertible, the first
column of $C_1$ is not 0. By multiplying $(C_1,\ldots,C_M)$ by a
unitary matrix on the left hand side, we may assume that only the
first component of the first column of $C_1$ is nonzero. Clearly, we
can also assume that the same is true for the first row of $C_1$.
Since $\r$ is irreducible, at least one of the first rows of the
$C_i$, $i>1$, must be nonzero. It follows that $\r$ is trivially
distillable. \epf

In view of Lemma \ref{le:Mx3reducibleNPT}, we have the following
corollary.
\bcr \label{cor:Mx3-NPT->dist} If $N=3$ we can remove the
irreducibility hypothesis from Proposition
\ref{pp:irreducibledistillable,commonkernel}.
\ecr

It is easy to construct examples of bipartite states having LFRP and
RFRP which can be shown to be distillable by
Proposition \ref{pp:irreducibledistillable,commonkernel}.
Clearly, Theorem~\ref{thm:distillation=nonFRP} cannot be applied
to such states.

\section{Separability criterion for bipartite states of rank 4}

The problem of deciding whether a state is separable has been shown
to be NP-hard \cite{gurvits03}, and hence the progress in solving
instances of this problem is significant for both the quantum
information theory and computer science. Most of the recent
contributions focus on the numerical methods with objective to
improve the efficiency in some special cases \cite{hm10}.

On the other hand the PPT condition, a celebrated decision criterion, is a necessary condition for separability of arbitrary $M\times N$
states. It is sufficient only if $MN\le6$. It is also sufficient for
some other classes of bipartite states. For instance, this is the
case when the rank, $R$, of $\r$ satisfies the inequality
$R\le\max(M,N)$. However, note that if $R<\max(M,N)$ then $\r$ is
NPT by result (C). Thus, the PPT condition forces the inequality
$\max(M,N)\le R$. If $\max(M,N)=R$, then the above result applies
and it remains to consider only the case $\max(M,N)<R$. We
infer that if $R\le3$ then $\r$ is separable if and only if it is
PPT, i.e., there are no PPT entangled states with $R\le3$. However,
when $R=4$ and $M=N=3$ such states do exist and the problem arises
to decide whether a given $3\times3$ PPT state of rank 4 is separable.
In this section we provide a simple answer to this question, and thereby obtain a criterion for separability of arbitrary bipartite
states of rank 4.

We start with an easy observation.
 \bl
\label{le:3x3rank4,rankC=1}
Let $\r$ be a $3\times N$ state such that, for some
$\ket{a}\in\cH_A$, $\rank \bra{a}\r\ket{a}=1$.
If $\r$ is NPT then it is distillable. If $\r$ is PPT and
$N=3$, then $\r$ is separable.
 \el
 \bpf
We shall prove both assertions at the same time.
We may assume that $\ket{a}=\ket{1}$. Consequently, we have
$\r=(C_1,C_2,C_3)^\dag\cdot(C_1,C_2,C_3)$, where the blocks $C_i$
are $R\times N$, $R$ is the rank of $\r$, and $C_1$ has rank 1.
After multiplying $(C_1,C_2,C_3)$ by a unitary matrix on the left
hand side and multiplying each $C_i$ by the same invertible matrix
on the right hand side, we may assume that
\begin{equation} \label{TriCi}
C_1=
 \left(\begin{array}{cc}
    1 & 0
 \\ 0 & 0
 \end{array}\right),
 \quad
 C_2=
 \left(\begin{array}{cc}
    x & u
 \\ C_{21} & C_{22}
 \end{array}\right),
 \quad
 C_3=
 \left(\begin{array}{cc}
    y & v
 \\ C_{31} & C_{32}
 \end{array}\right),
\end{equation}
where the first blocks are $1\times1$. If $u\ne0$ or $v\ne0$ then
$\r$ is trivially distillable and so NPT. Thus we may assume
that $u=v=0$. By applying an ILO on $\cH_A$, we may also assume
that $x=y=0$. After switching the parties A and B, $\r$ becomes
reducible. Hence, if $\r$ is NPT it is distillable by
Lemma~\ref{le:Mx3reducibleNPT}, and otherwise it is separable
by Corollary~\ref{cor:reducible=SUMirreducible}.
 \epf

The following proposition is the crucial step in the proof of
our separability criterion.
 \bpp
 \label{pp:3x3rank4,productstate}
Let $\r$ be a $3\times 3$ state of rank 4 containing at least one
product state in its range. Then $\r$ is either distillable or
separable. (In the former case it is NPT and in the latter PPT.)
Equivalently, $\r$ cannot be PPT and entangled.
 \epp
 \bpf
If $\r$ is reducible, then the assertion follows from
Lemma~\ref{le:Mx3reducibleNPT} and
Corollary~\ref{cor:reducible=SUMirreducible}.
Hence, we may assume that $\r$ is irreducible.

We can write $\r$ as $\r=\sum_{i=1}^4 \ketbra{\psi_i}{\psi_i}$,
where $\ket{\psi_1}$ a product state. This gives the factorization
$\r=(C_1,C_2,C_3)^\dag\cdot(C_1,C_2,C_3)$ with the blocks $C_i$ of
size $4\times3$. By applying an ILO, we may assume that
 $$ C_1=
 \left(\begin{array}{cc}
    1 & 0
 \\ C_{11} & C_{12}
 \end{array}\right),
 \quad
 C_2=
 \left(\begin{array}{cc}
    0 & 0
 \\ C_{21} & C_{22}
 \end{array}\right),
 \quad
 C_3=
 \left(\begin{array}{cc}
    0 & 0
 \\ C_{31} & C_{32}
 \end{array}\right),
 $$
where the blocks $C_{i2}$ are of size $3\times2$. If the projected
state $\r':=(C_{12},C_{22},C_{32})^\dag\cdot(C_{12},C_{22},C_{32})$
is entangled, then both $\r'$ and $\r$ are distillable. So, we may assume that $\r'$ is separable.
Since a separable state of rank 3 is a sum of 3 product states
(see Proposition \ref{pp:PPTrankN}),
we may assume that $C_{i2}=D_iC$, $i=1,2,3$, where $C$ is a
$3\times2$ matrix and the $D_i$ are diagonal matrices.
If $D_2$ and $D_3$ are linearly dependent, then we can assume that
one of them is 0 and so Lemma~\ref{le:3x3rank4,rankC=1} implies that
$\r$ is distillable or separable.
Thus we may assume that $D_2$ and $D_3$ are linearly independent.
They remain independent after removing one of
the rows in each of them, say the first row. By using an ILO on
system A, we may assume that the diagonal entries of $D_1$, $D_2$
and $D_3$ are $(d_1,0,0)$, $(d_2,1,0)$ and $(d_3,0,1)$,
respectively. The matrices $C_i$ now have the form
 $$ C_1=
 \left(\begin{array}{ccc}
    1 & 0  & 0
  \\ u_1 & d_1 \a & d_1 \b
  \\ u_2 & 0 & 0
  \\ u_3 & 0 & 0
 \end{array}\right),
 \quad
 C_2=
 \left(\begin{array}{ccc}
    0 & 0  & 0
  \\ v_1 & d_2 \a & d_2 \b
  \\ v_2 & \gamma & \delta
  \\ v_3 & 0 & 0
 \end{array}\right),
 \quad
 C_3=
 \left(\begin{array}{ccc}
    0 & 0  & 0
  \\ w_1 & d_3 \a & d_3 \b
  \\ w_2 & 0 & 0
  \\ w_3 & \epsilon & \zeta
 \end{array}\right).
 $$
If $d_1=0$ or $\a=\b=0$, we can again use
Lemma~\ref{le:3x3rank4,rankC=1}.
Thus, by using an ILO on system B, we may assume that $d_1=\a=1$
and $\b=0$, as well as that $u_1=0$. One of $\delta$
and $\zeta$ must be nonzero, and so we may assume that, say
$\delta=1$, as well as $v_2=\gamma=0$.

For convenience, denote by $\r_{ij}$ the projected state
$(C_i,C_j)^\dag\cdot(C_i,C_j)$, $1\le i<j\le3$.
If $u_2\ne0$ then $\r_{12}$ is trivially distillable, and so we may
assume that $u_2=0$. Now we have
 $$ C_1=
 \left(\begin{array}{ccc}
     1 & 0 & 0
  \\ 0 & 1 & 0
  \\ 0 & 0 & 0
  \\ u_3 & 0 & 0
 \end{array}\right),
 \quad
 C_2=
 \left(\begin{array}{ccc}
     0 & 0  & 0
  \\ v_1 & d_2 & 0
  \\ 0 & 0 & 1
  \\ v_3 & 0 & 0
 \end{array}\right),
 \quad
 C_3=
 \left(\begin{array}{ccc}
     0 & 0  & 0
  \\ w_1 & d_3 & 0
  \\ w_2 & 0 & 0
  \\ w_3 & \epsilon & \zeta
 \end{array}\right).
 $$
By subtracting from $C_2$ and $C_3$ suitable multiples of $C_1$, we
may assume that
$$ C_2=\left(\begin{array}{ccc}
     -d_2 & 0  & 0
  \\ v_1 & 0 & 0
  \\ 0 & 0 & 1
  \\ v'_3 & 0 & 0
 \end{array}\right), \quad
 C_3=
 \left(\begin{array}{ccc}
     -d_3 & 0  & 0
  \\ w_1 & 0 & 0
  \\ w_2 & 0 & 0
  \\ w'_3 & \epsilon & \zeta
 \end{array}\right). $$
If $v_1\ne0$ then $\r_{12}$ is trivially distillable, so we may
assume that $v_1=0$.

Assume that $\zeta=0$. Since $\r$ is irreducible, we must have
$w_2\ne0$ and so $\rho_{23}$ is trivially distillable.
Assume now that $\zeta\ne0$.
If $u_3\ne0$ then $\rho_{13}$ is trivially distillable,
so we may assume that $u_3=0$. If $w_1\ne0$ then it is easy to see
that the state $\r_{13}$ is distillable. Indeed, for that purpose
we may assume that $\epsilon=0$ and then $\rho_{13}$ becomes
trivially distillable. Thus we may assume that also $w_1=0$, and so
 $$ C_1=\left(\begin{array}{ccc}
     1 & 0 & 0
  \\ 0 & 1 & 0
  \\ 0 & 0 & 0
  \\ 0 & 0 & 0
 \end{array}\right), \quad
 C_2=\left(\begin{array}{ccc}
     -d_2 & 0  & 0
  \\ 0 & 0 & 0
  \\ 0 & 0 & 1
  \\ v'_3 & 0 & 0
 \end{array}\right), \quad
 C_3= \left(\begin{array}{ccc}
     -d_3 & 0  & 0
  \\ 0 & 0 & 0
  \\ w_2 & 0 & 0
  \\ w'_3 & \epsilon & \zeta
 \end{array}\right). $$

Since $\r$ is irreducible, we must have $\epsilon\ne0$. If
$v'_3\ne0$ then $\r_{23}$ is trivially distillable. Thus we assume
that $v'_3=0$. If $w_2=0$ then $\r$ is separable. Thus we assume
that $w_2\ne0$. Then we examine the projected state $\r_{23}$.
As $\epsilon\ne0$, we can kill $w'_3$ and $\zeta$ in $C_3$ to obtain
a trivially distillable state. This completes the proof.
 \epf

Note that the above proof actually shows that the NPT states
treated in this proposition are in fact 1-distillable.

We now present our separability criterion for bipartite states of
rank 4.
 \bt
 \label{thm:PPTrank4}
A bipartite state of rank 4 is separable if and only if it is
PPT and its range contains at least one product state.
 \et

\bpf
The conditions are obviously necessary. To prove the sufficiency,
let $\r$ be an $M\times N$ bipartite state of rank 4. By the
result (C), we must have $M,N\le4$. If $\max(M,N)=4$ then
Proposition \ref{pp:PPTrankN} shows that $\r$ is separable.
In view of the Peres-Horodecki criterion, it remains to
consider the case $M=N=3$. Hence, we can invoke Proposition
\ref{pp:3x3rank4,productstate} to complete the proof.
\epf

As an open question we ask whether Proposition
\ref{pp:3x3rank4,productstate}
can be extended to some other classes of PPT states, such as
$2\times4$ PPT entangled states of rank 5. Such states do exist
\cite{horodecki97} (see also  \cite[Eq. (6)]{agk09}).

Hence to decide whether a $3\times3$ PPT state $\r$ of rank 4 is
separable, one just has to check whether there is a product state in
the range of $\r$. This is numerically operational because the
dimension is low. For example, one can refer to the discussion in
\cite[Sec. IV]{hlv00}. Here we present the analytical approach to
this problem. We only need to consider the linear
combination of four linearly independent $3\times3$ matrices
$A=[a_{ij}]$, $B=[b_{ij}]$, $C=[c_{ij}]$, $D=[d_{ij}]$.
The problem is to decide whether the variables $w,x,y,z$ can be
chosen so that the matrix
$$ E=wA+xB+yC+zD $$
has rank 1. This is equivalent to the requirement that each
$2\times2$ minor of $E$ vanishes. Mathematically, we have a set of nine equations
 \bea
 \label{ea:3x3}
 &&(w a_{ij}+xb_{ij}+yc_{ij}+zd_{ij})
(wa_{kl}+xb_{kl}+yc_{kl}+zd_{kl}) \nonumber\\
 &&-(wa_{il}+x b_{il}+yc_{il}+zd_{il})
(wa_{kj}+xb_{kj}+yc_{kj}+zd_{kj})=0,
 \eea
with $i<k$ and $j<l$.
Thus $\rank E=1$ if and only if these equations have a nonzero
solution for $w,x,y,z$.

It is instructive to take a look at the following example
which requires more algebraic background.

\begin{example}
{\rm
Let us consider the simpler problem of deciding whether a
2-dimensional subspace $V$ of a $2\ox 3$ system contains
a product vector. Let $\{\ket{a},\ket{b}\}$ be an arbitrary
basis of $V$. Thus
$$ \ket{a}=\sum_{i,j} a_{ij}\ket{i,j}, \quad
\ket{b}=\sum_{i,j} b_{ij}\ket{i,j}, $$
where $i=1,2$ and $j=1,2,3$. Let us form the $2\times6$ matrix
from the components of these two states:
$$
P=\left( \begin{array}{cccccc}
a_{11} & a_{12} & a_{13} & a_{21} & a_{22} & a_{23} \\
b_{11} & b_{12} & b_{13} & b_{21} & b_{22} & b_{23}
\end{array} \right). $$
Denote by $P^{ij}$, $i<j$, the $2\times2$ submatrix
of $P$ made up from the $i$th and $j$th columns.
The Pl{\"u}cker coordinates of $V$ are the 15 determinants
$p_{ij}=\det P^{ij}$. If we change the basis, the Pl{\"u}cker
coordinates of $V$ will be changed only by an overall factor.
We point out that the Pl{\"u}cker coordinates are algebraically
dependent. For instance, we have
$p_{12}p_{34}-p_{13}p_{24}+p_{14}p_{23}=0$.
A 2-dimensional subspace $V$ can be viewed as a point
of the so called Grassmann variety (or Grassmannian) $G_{2,6}$.
The set of the points $V$ which contain a product vector form a
closed subvariety of this Grassmannian. By using the known facts
about the incidence varieties as presented in \cite{h92},
one can show that this subvariety is in fact a hypersurface,
i.e., it is given by just one algebraic equation in the
Pl{\"u}cker coordinates. This equation can be computed explicitly,
it is given by a homogeneous polynomial of degree 3:
\begin{eqnarray*}
&& 2p_{12}p_{34}p_{56}+p_{12}p_{26}p_{46}+p_{13}p_{15}p_{56}+
p_{23}p_{24}p_{46}+p_{13}p_{35}p_{45} \\
&& -p_{13}p_{25}p_{46}-
p_{13}p_{24}p_{56}-p_{12}p_{35}p_{46}-p_{12}p_{16}p_{56}-
p_{23}p_{34}p_{45}=0.
\end{eqnarray*}
Thus, $V$ contains a product vector if and only if its
Pl{\"u}cker coordinates satisfy this equation. The necessity can
be easily verified by setting $a_{ij}=\a_i\beta_j$, computing
the $p_{ij}$, and then verifying that the above equation is
identically satisfied. For sufficiency, due to the fact that the
$p_{ij}$ are algebraically dependent, one must verify that the
left hand side of the equation is not identically zero. For
instance, if we set $a_{11}=a_{22}=b_{12}=b_{23}=1$ and all
other $a_{ij}$ and $b_{ij}$ set equal to 0, then the left hand
side of the equation is equal to $-1$.
}
\end{example}

Similar polynomial equation exists for 3-dimensional subspaces $V$ of
a $2\ox 4$ system. The polynomial is again homogeneous, but now has
degree 4, and it is an integer linear combination of 149 monomials
in the Pl{\"u}cker coordinates $p_{ijk}$, $1\le i<j<k\le8$. This
equation is given in the appendix.
In the case of 4-dimensional subspaces of a $3\ox 3$ system
again there is such an equation, but so far we were not able
to compute it explicitly.

We note that Proposition~\ref{pp:3x3rank4,productstate} is in
agreement with a conjecture proposed in \cite{lms10}, where the
authors ask whether all $3\times3$ PPT entangled states of rank
4 are equivalent via stochastic LOCC to the PPT entangled states
arising from UPB \cite{bdm99}, i.e., the states
$V_A \ox V_B~(I_9-\sum^5_{i=1} \proj{a_i,b_i})~V_A^\dg \ox V_B^\dg$
where $V_A,V_B$ are invertible and $\ket{a_i,b_i}$ are the five
members of a UPB.
Indeed, Proposition~\ref{pp:3x3rank4,productstate} shows that
$\cR(\r)$ contains no product state, and so provides evidence in
support of the conjecture.

\section{\label{sec:Applicat}
Some applications}

In this section we apply Theorem~\ref{thm:distillation=MxNrankN} to
study the quantum correlations inside tripartite pure states. We also
discuss its meaning in terms of quantum discord \cite{oz02}. We
denote by $d_A,d_B,d_C$ the dimensions of the Hilbert spaces
$\cH_A,\cH_B,\cH_C$, respectively.

As a bipartite state is always a reduced density operator of a
tripartite pure state, say $\r=\proj{\psi}$, we can use Theorem
\ref{thm:distillation=MxNrankN} to characterize the quantum
correlation by means of distillable entanglement. For example, we
ask when are $\r_{AB}$ and $\r_{AC}$ simultaneously undistillable or
bound entangled and which states $\ket{\ps}$ satisfy such a
property. The question was first studied in 1999 by Thapliyal
\cite{thapliyal99}, who showed that an $N$-partite pure state
$\ket{\ps}$ has fully $(N-1)$-partite separable states if and only
if $\ket{\ps}$ is a GHZ state up to local unitary. Based on
Theorem~\ref{thm:distillation=MxNrankN}, we solve completely this
problem for the tripartite system. We need a lemma to prove our
theorem.

\begin{lemma}
 \label{le:pptpptX=SSX}
For a tripartite pure state $\r=\proj{\ps}$, the bipartite reduced
density operators $\r_{AB}$ and $\r_{AC}$ are PPT if and only if
$\ket{\ps}=\sum^d_{i=1} \ket{a_i} \ket{ii}$
up to local unitary operations.
\end{lemma}
\begin{proof}
The sufficiency is obvious, e.g., the reduced state
$\r_{AB}=\tr_C\r=\sum_i\ketbra{a_i,i}{a_i,i}$ is separable.
Let us show the necessity. Suppose
$\r_{AB}^{\G} \geq 0$ and $\r_{AC}^{\G} \geq 0$. Recall that
$d_B=\rank\r_{AC}$ and $d_C=\rank\r_{AB}$. If follows from
result (C) that $d_B \geq \max(d_A, d_C)$ and
$d_C \geq \max(d_A, d_B)$. So we have $d:=d_B=d_C\geq d_A$.
According to Proposition~\ref{pp:PPTrankN}, the states
$\r_{AB}$ and $\r_{AC}$ are separable and
$\r_{AB} = \sum^d_{i=1} \proj{a_i, b_i}$,
where the states $\ket{b_i}$ span $\mathcal{H}_B$.
So, $\ket{\ps}= \sum^d_{i=1} \ket{a_i,b_i,i}$ up to local unitary
operations.

We choose a subset $\{\ket{g_1},\ldots,\ket{g_m}\}$ of
$\{\ket{a_1},\ldots,\ket{a_d}\}$ such that (i) $\ket{g_i}$ and
$\ket{g_j}$ are not parallel if $i\ne j$ and (ii) each $\ket{a_i}$
is parallel to some $\ket{g_j}$. For each $k\in\{1,\ldots,m\}$
let $S_k$ be the set of all $i$ such that $\ket{a_i}$ is parallel
to $\ket{g_k}$. Thus the sets $S_1,\ldots,S_m$ form a partition
of the set $\{1,\ldots,d\}$.

The state $\r_{AC}$ is also separable and has rank $d$.
It follows from Proposition~\ref{pp:PPTrankN} again that
$\r_{AC}=\sum^d_{i=1}\proj{e_i, f_i}$. Since the product states
$\ket{a_i,i}$ span $\cR(\r_{AC})$, there is an
invertible matrix $[c_{ij}]$ of order $d$ such that
$$ \ket{e_i, f_i}=\sum^d_{j=1}c_{ij}\ket{a_j,j}. $$
Since the $\ket{j}_C$ form an o.n. basis of $\cH_C$, this equation
implies that whenever coefficients $c_{ij}$ and $c_{il}$ are nonzero
then $\ket{a_j}$ and $\ket{a_l}$ must be parallel.
In other words there is a $k$ such that $\ket{e_i,f_i}$ is a
linear combination of $\ket{g_k,j}$ with $j\in S_k$. For each
$k\in\{1,\ldots,m\}$ let $\cH_k$ be the subspace of $\cH_C$
spanned by the $\ket{j}$ with $j\in S_k$, and let $F_k$ be the set
of all $i$ such that $\ket{e_i,f_i}$ is a linear combination of
$\ket{g_k,j}$ with $j\in S_k$. Clearly, the sets $F_1,\ldots,F_m$
also form a partition of $\{1,\ldots,d\}$.
It is easy to see that, for each $k$, the sets $F_k$ and $S_k$
have the same cardinality, say $d_k$.

By using the spectral decomposition, we can rewrite the state
$\s_k=\sum_{i\in F_k} \proj{e_i,f_i}$ as
$\s_k=\sum_{i=1}^{d_k} \proj{g_k,f_{k,i}}$,
where the states $\ket{f_{k,i}}$, $1\le i\le d_k$, form an
orthogonal basis of $\cH_k$. By taking the sum of all $\s_k$s,
and renaming all product states $\ket{g_k,f_{k,i}}$ as
$\ket{e_j',f_j'}$, we obtain
$$ \r_{AC} =  \sum^d_{j=1} \proj{e_j',f_j'}, $$
where the states $\ket{f_j'}$ are pairwise orthogonal. The
purification of this state, namely $\ket{\ps}$, reads
$\ket{\ps} = \sum^d_{i=1} \ket{e_i'} \ket{ii}$,
up to local unitary operations. This proves the necessity.
\end{proof}

 \bt
 \label{thm:nonnonX=SSX}
For a tripartite pure state $\r=\proj{\ps}$, the reduced states
$\r_{AB}$ and $\r_{AC}$ are undistillable if and only if
$\ket{\ps}=\sum^d_{i=1} \ket{a_i}\ket{ii}$
up to local unitary operations.
 \et
 \bpf
The sufficiency readily follows from the separability of $\r_{AB}$
and $\r_{AC}$. Let us show the necessity. As in the above proof we
have $d_B = d_C \geq d_A$. By
Theorem~\ref{thm:distillation=MxNrankN}, the states $\r_{AB}$ and
$\r_{AC}$ are PPT. Then the assertion follows from
Lemma~\ref{le:pptpptX=SSX}.
 \epf

 \bcr
 \label{cor:nonnonnon=SSS}
For a tripartite state $\ket{\ps}$, all three bipartite reduced
states are undistillable if and only if $\ket{\ps}$ is a generalized
GHZ state, i.e., $\ket{\ps}=\sum^d_{i=1} a_i\ket{iii}$ up to local
unitary operations.
 \ecr

Theorem~\ref{thm:nonnonX=SSX} reveals a new constraint for the
distillable entanglement: When a tripartite pure state has two
undistillable reduced density operators, then they have to be
separable. In other words, when a tripartite pure state has a bound
entangled reduced state, such as a PPT entangled state, then the
other two reduced states must be NPT and distillable. So regardless
of whether bound entanglement is PPT, or NPT as conjectured in
\cite{dss00}, it is not a generic property shared by two parties of
a tripartite pure state. The question whether a mixed tripartite
state may have two bound entangled reduced states is still an open
problem.

It is interesting that there is no counterpart to
Corollary \ref{cor:nonnonnon=SSS} for mixed states. For example, we
consider the three-qubit PPT entangled state
$\r=I-\sum^4_{i=1}\proj{\ps_i}$, where the normalized
$2\ox 2\ox 2$ unextendible product basis (UPB) reads \cite{bdm99,bravyi04}
 \bea
 \ket{\ps_1} &=& \ket{1,1,1},
 \nonumber\\
 \ket{\ps_2} &=& \ket{2,b,c},
 \nonumber\\
 \ket{\ps_3} &=& \ket{a,2,c^{\bot}},
 \nonumber\\
 \ket{\ps_4} &=& \ket{a^{\bot},b^{\bot},2},
 \eea
and $\ket{a},\ket{a^{\bot}}, \ket{b},\ket{b^{\bot}}$ and
$\ket{c},\ket{c^{\bot}}$ are orthornormal, respectively. All three
bipartitions of this state, i.e., $\r_{A:BC},\r_{B:AC},\r_{C:AB}$
are separable. Hence all three bipartite reduced states are
separable too. However the state $\r$ is PPT entangled and we cannot
distill any (bipartite) pure entanglement for quantum information
tasks. This is essentially different from Corollary
\ref{cor:nonnonnon=SSS}. From sufficiently many copies of
generalized GHZ state, one can asymptotically generate a standard
GHZ state $\ket{000}+\ket{111}$ based on the BBPS protocol
\cite{bbp96PRA}. Therefore besides the distillable entanglement, we
need additional parameters to characterize the quantum correlation
inside mixed tripartite states.

On the other hand Theorem~\ref{thm:nonnonX=SSX} can be better
understood via quantum discord \cite{oz02}, which measures the
bipartite quantum correlation beyond entanglement. Indeed, quantum
discord is larger than zero for many separable states, and it is
equal to zero iff the separable state is diagonal in one of the
systems, e.g., $\sum_i \r_i \ox \proj{i}$. In this state any
measurement on system A will lead to pairwise commuting Hermitian
operators on system B, which is a basic property of classical
mechanics. In this sense, the system $B$ is {\em classical} and it
has no quantum features. Furthermore the converse statement is also
true. That is, the state has the form $\sum_i \r_i\ox \proj{i}$ if
the system B is classical \cite{ccm10}. For simplicity we say that
the state $\r_{AB}$ is classical when the system $B$ is classical.
By using Theorem \ref{thm:nonnonX=SSX} and the definition of quantum
discord, we have

 \bl
 \label{le:two,threeclassical}
For a tripartite pure state $\ket{\ps}$, the reduced states
$\r_{AB}$ and $\r_{AC}$ are undistillable if and only if they have
zero quantum discord. In particular, $\ket{\ps}$ is a generalized
GHZ state if and only if all reduced states have zero quantum
discord.
 \el

Lemma \ref{le:two,threeclassical} shows that the collective
undistillability of reduced states in a tripartite pure state will
lead to the disappearance of quantum correlation among them. It is
unknown whether this is also true for mixed tripartite states.

To grasp the meaning of this result more intuitively, we consider a
practical case in which Alice and Bob are correlated in a state
$\r_{AB}$, and there is a third party, Charlie, from the
environment. Suppose they are in a pure state $\ket{\ps}$ and share
sufficiently many copies of $\ket{\ps}$. First, when at most one
system is classical, say Alice, we can distill at least two
bipartite reduced states by Theorem~\ref{thm:nonnonX=SSX} and
Lemma~\ref{le:two,threeclassical}. Second, when two systems are
classical, say Bob and Charlie, then there may be either one or two
distillable bipartite reduced states. The former tripartite state is
just $\ket{\ps}=\sum^d_{i=1} \ket{a_i}\ket{ii}$ of
Lemma~\ref{le:two,threeclassical}, while the latter corresponds to
the purification of a fully classical state of the form
$\sum_{ij}a_{ij}\proj{ij}$ up to local unitary \cite{ccm10}. In
other words, the state $\sum_{ij}
\sqrt{a_{ij}}\ket{\ph_{ij}}\ket{ij}$ has classical systems $B,C$ in
the same reduced state. Meanwhile, we have distillable $\r_{AB}$ and
$\r_{AC}$ when $d_A>\max(d_B,d_C)$ in view of result (C). Finally,
when all three systems are classical then the state $\ket{\ps}$ is a
generalized GHZ state. One cannot distill entanglement between any
two parties. This indicates that the less classical systems there
are in a composite system, the more distillable reduced states there
are. This provides a connection between quantum discord and
distillable entanglement, where the former is viewed as a kind of
quantum correlation rather than entanglement
measure \footnotemark.
\footnotetext{A connection between
quantum discord and the entanglement of formation, which is another
important entanglement measure, has been found by several authors,
see e.g., quant-ph/1007.1814, quant-ph/1006.4727.}

\section{\label{sec:distilling3x3rank4}
Distillation of some $3\times3$ states of rank 4}

So far we have identified some distillable entangled states, e.g.,
Theorem~\ref{thm:distillation=MxNrankN} shows that $M\times N$ NPT
states of rank $N$ are 1-distillable. A natural question then
arises: can we construct some more complex distillable states
with specified dimension and rank? In this section we focus on the
simplest nontrivial case, namely $3\times3$ states of rank 4.
In particular we will show that the checkerboard states \cite{dzd11}, which generalize the celebrated Bru{\ss}-Peres family, are
distillable if and only if they are NPT. Some of
the results are also extendible to higher dimensions.

We shall now examine the so called {\em checkerboard states},
i.e., the states $\r$ acting on a $3\ox 3$ system and defined by
Eq.~(\ref{jed:ro-psi}) with $M=N=3$, $R=4$, and the $\ket{\psi_i}$
given by:
\begin{eqnarray*}
\ket{\psi_1} &=& \ket{1}(a\ket{1}+d\ket{3})+\ket{2}(c\ket{2})
+\ket{3}(b\ket{1}+e\ket{3}), \\
\ket{\psi_2} &=& \ket{1}(g\ket{2})+\ket{2}(f\ket{1}+i\ket{3})
+\ket{3}(h\ket{2}), \\
\ket{\psi_3} &=& \ket{1}(j\ket{1}+m\ket{3})+\ket{2}(l\ket{2})
+\ket{3}(k\ket{1}+n\ket{3}), \\
\ket{\psi_4} &=& \ket{1}(q\ket{2})+\ket{2}(p\ket{1}+s\ket{3})
+\ket{3}(r\ket{2}). \\
\end{eqnarray*}
The parameters $a,b,\ldots,s$ are arbitrary complex numbers.
Each pure state $\ket{\psi_i}$ is written above in the form
$\ket{\psi_i}=\sum_j\ket{i}\ket{\psi_{ij}}$. By using these
$\ket{\psi_{ij}}$ and Eq. (\ref{jed:mat-Cj}) we have
$\r=(C_1,C_2,C_3)^\dag\cdot(C_1,C_2,C_3)$, where the complex
conjugates of the blocks $C_i$ are given by:
$$ C_1^*=\left(\begin{array}{ccc}
a&0&d\\0&g&0\\j&0&m\\0&q&0\end{array}\right),\quad
C_2^*=\left(\begin{array}{ccc}
0&c&0\\f&0&i\\0&l&0\\p&0&s\end{array}\right),\quad
C_3^*=\left(\begin{array}{ccc}
b&0&e\\0&h&0\\k&0&n\\0&r&0\end{array}\right). $$

Generically, these states $\r$ have rank 4 and $\cR(\r)$ contains
no product state (see \cite{dzd11}) and, consequently, $\r$ is
entangled. One can find in the same paper two concrete
examples of NPT checkerboard states $\r$ which are distillable.
In the next theorem we prove that this feature is shared by all
NPT checkerboard states.
 \bt
 \label{thm:checkboard}
 All NPT checkerboard states are distillable.
 \et
 \bpf
Let $\r$ be an NPT checkerboard state. We have to show that $\r$ is
distillable. This is certainly the case if one of the local ranks of
$\r$ is less than 3. Hence we may assume that $\r$ is a $3\times3$
state. We divide the proof into two steps. First, we eliminate all
but five parameters from the matrices $C_i^*$. Second, we
analytically investigate the partial transpose of various $2\times2$
and $2\times3$ projected states of $\r$. If any of them is entangled
then $\r$ is distillable and so we can dismiss such cases.

Let us carry out the first step. By an argument similar to one in
the beginning of the proof of
Proposition~\ref{pp:3x3rank4,productstate}, we can assume that
$b=e=j=m=0$. By using an ILO, we can also assume that $a=g=1$ and
$d=q=0$. Thus we have
$$ C_1^*=\left(\begin{array}{ccc}
 1&0&0\\0&1&0\\0&0&0\\0&0&0\end{array}\right),\quad
 C_2^*=\left(\begin{array}{ccc}
 0&c&0\\f&0&i\\0&l&0\\p&0&s\end{array}\right),\quad
 C_3^*=\left(\begin{array}{ccc}
 0&0&0\\0&h&0\\k&0&n\\0&r&0\end{array}\right). $$
If $i\ne0$ then $\r$ is trivially distillable, and so we assume that
$i=0$. If $s=0$ then $\r$ is distillable by
Corollary~\ref{cor:Mx3-NPT->dist}, so we may assume that $s\ne0$.
Similarly, we may assume that $n\ne0$. By multiplying $C^*_2$ by
$1/f$, we may assume that $f=1$. By multiplying the third columns of
$C^*_i$ by $1/s$, we may also assume that $s=1$. By multiplying
$C^*_3$ by $1/n$, we may assume that $n=1$. By using an ILO we can
also assume that $p=0$. We now have
 $$ C_1^*=\left(\begin{array}{ccc}
 1&0&0\\0&1&0\\0&0&0\\0&0&0\end{array}\right),\quad
 C_2^*=\left(\begin{array}{ccc}
 0&c&0\\1&0&0\\0&l&0\\0&0&1\end{array}\right),\quad
 C_3^*=\left(\begin{array}{ccc}
 0&0&0\\0&h&0\\k&0&1\\0&r&0\end{array}\right).$$
This completes the first step.

Now we carry out the second step. Let $\s=(V \ox I) ~\r~ (V^\dg \ox
I)$, where $V^\dag=(x\ket{1}+\ket{3})\bra{1}+\proj{2}$ and $x$ is a
complex parameter. Assume that $\s$ is PPT for all $x$. After
permuting simultaneously the rows and columns of $\s^\G$, we obtain
a direct sum $A \oplus B$ of two $3\times3$ Hermitian matrices.
Since $\s$ is PPT, we must have
\begin{eqnarray*}
\det A &=& |x|^2(|c|^2+|l|^2-|r|^2)-|x+h^*-kr^*|^2\ge 0, \\
\det B &=& |x+h^*|^2-|cx+lk^*|^2+|r|^2-|l|^2 \ge 0.
\end{eqnarray*}
By setting $x=0$, the first inequality gives $h=rk^*$. After
eliminating $h$, the above inequalities imply that $|c|=1$ and
$|l|=|r|$, and finally that $l=cr^*k/k^*$ if $k\ne0$. By using these
equalities, a computation shows that $\r$ is PPT which is a
contradiction. Hence $\s$ must be NPT for some $x$, and so this
particular $\s$ and $\r$ are distillable. One can similarly and more
easily handle the case $k=0$. This completes the proof.
 \epf

\section{\label{sec:full-rank}
More on full-rank properties}

Theorem~\ref{thm:distillation=nonFRP} gives a new theoretical tool,
the full-rank criterion, for detection of bipartite states
distillable under LOCC. This criterion is similar to the well-known
reduction criterion~\cite{hh99}, for both criteria ensure the
distillability of the states which violate them. It is easy to show
that these two criteria are incomparable. First, the distillable
$3\times3$ antisymmetric Werner state $\r_{as}$ is detected by the
full-rank criterion but not by reduction criterion. Second, there
are 1-distillable states $\r$ which can be detected by reduction
criterion but not by the full-rank criterion. An example is the
$2\times2$ entangled state $\r=\proj{\psi_1}+\proj{\psi_2}$, where
$\ket{\psi_1}=\sqrt{p/2}(\ket{00}+\ket{11})$ and
$\ket{\psi_2}=\sqrt{(1-p)/2}(\ket{00}-\ket{11})$ with $0<p<1$ and
$p\ne1/2$. It follows from Proposition \ref{pp:Mx2} below that all
$2\times2$ states have both LFRP and RFRP. Third, some distillable
states are detected by both criteria such as $M\times N$ states of
rank $<N$; in particular the pure entangled states.

Since the states violating either criterion are distillable, they
must be NPT. In other words, the PPT states satisfy both reduction
criterion and the full-rank criterion.

Next we consider the states with RFRP. Let $\r$ be an $M\times N$
state of rank $R$ written as in Eq. (\ref{ro}) where the $C_i$ are
$R\times N$ matrices. The matrix of $\r$ is the block matrix
$(C_i^\dag C_j)$, $i,j=1,\ldots,M$, where each block is an
$N\times N$ matrix. If $\rank(C_i)<N$, then the rank of the
submatrix $\left(C_i^\dag C_1,\ldots,C_i^\dag C_M \right)$ is $<N$.
Hence, if $\rank(C_i)<N$ for all $i$, then $R\le M(N-1)$. We
conclude that if the rank of $\r$ exceeds $M(N-1)$ then $\r$ must
have RFRP (in fact one of the $C_i$s must have rank $N$).

On the other hand, if $R<N$ then $\r$ violates RFRP. This follows
from the fact that $\rank(\bra{x}\r\ket{x}) \le R$ for all
$\ket{x}\in\cH_A$. Moreover, the $N\times N$ states $\r$ in the
antisymmetric space violate RFRP and some of them have rank
as large as $N(N-1)/2$. To prove this, note that the states
$\ket{\pi_{ij}}=\ket{ij}-\ket{ji}$, $1\le i<j\le N$, form a basis of
the antisymmetric space $\cA$. We have $\r=\sum_k \proj{\psi_k}$,
where $\ket{\psi_k}\in\cA$ are non-normalized states. Let
$\ket{x}\in\cH_A$ be a nonzero vector. For $i<j$ we have
$$ \braket{\pi_{ij}}{x,x}=\braket{i,j}{x,x}-\braket{j,i}{x,x}=0. $$
Since each $\ket{\psi_k}$ is a linear combination of the
$\ket{\pi_{ij}}$, it follows that $\ket{x,x}\in\ker\r$.
Consequently, $\r$ violates RFRP.

For convenience, we collect the above results in a proposition.
 \bpp
Let $R$ be the maximum rank of $M\times N$ states which violate RFRP.
Then $R\le M(N-1)$ and, if $M=N$, $R\ge N(N-1)/2$.
 \epp

How can we verify whether $\r$ has RFRP? To answer this question,
let us write $\r$ as in Eq. (\ref{ro}), where the $C_i$ are
$R\times N$ matrices and $R=\rank(\r)$. For
$\ket{x}_A=\sum_k \xi_k\ket{k}$, we have
$$ \bra{x}\r\ket{x}=\left(\sum_k \xi_k C_k\right)^\dag \cdot
\left(\sum_k \xi_k C_k\right), $$
and so $\rank(\bra{x}\r\ket{x})=\rank(\sum_k \xi_k C_k)$.
For small $N$ the answer to our question can be obtained simply
by computing all $N\times N$ minors of the matrix $\sum_k\xi_k C_k$.

In general, $\r$ violates RFRP if and only if the space of
$R\times N$ matrices spanned by $C_1,\ldots,C_M$ contains no
matrix of rank $N$. This is certainly the case if $R<N$. The above
problem is related to the still open problem of
Edmonds \cite{edmonds67}, from theoretical computer science, which
asks to decide whether a given linear subspace of complex
$M\times N$ $(M\le N)$ matrices contains a matrix of rank $M$.

A similar test is valid for the LFRP. For
$\ket{y}_B=\sum_l \eta_l\ket{l}$, we have
$\bra{y}\r\ket{y}=\left(\bra{y}C_i^\dag C_j\ket{y}\right)$,
where the right hand member is an $M\times M$
matrix with the indicated entries. It follows that the rank of
$\bra{y}\r\ket{y}$ is equal to the rank of the $R\times M$ matrix
$\left(C_1\ket{y},\ldots,C_M\ket{y}\right)$. Let $K_i$,
$i=1,\ldots,N$, be the $R\times M$ matrix defined as follows:
the $j$th column of $K_i$ is the $i$th column of $C_j$. Then $\r$
violates LFRP if and only if the space of $R\times M$ matrices
spanned by $K_1,\ldots,K_N$ contains no matrix of rank $M$.

It is not hard to construct examples of states which possess
LFRP but violate RFRP.
In the next proposition we consider the RFRP for $M\times2$ states.

\bpp \label{pp:Mx2}
All $M\times2$ states have RFRP. Equivalently, all $2\times N$
states have LFRP.
\epp

\bpf
The equivalence is clear since we can interchange the parties
A and B. We shall prove only the first assertion.

The proof is by contradiction. Let us assume that there exists an
$M\times2$ state $\r$ which violates RFRP. Then the rank, $R$, of
$\r$ must be at least two. We have
$$ \r=\sum_{i,j=1}^M \ketbra{i}{j}\ox C_i^\dag C_j, $$
where $C_i$ are $R\times2$ matrices. As $\rank(\r_A)=M$,
$C_i\ne0$ for each index $i$. Since $\r$ violates RFRP,
the matrix $\sum_i \xi_iC_i$ must have rank less than two for
arbitrary choice of complex numbers $\xi_i$. It follows that
each $C_i$ must have rank one and so
$C_i=\ketbra{v_i}{\phi_i}$, where
$\ket{\phi_i}=\a_i\ket{1}+\b_i\ket{2}\in\cH_B$ and $\ket{v_i}$
are nonzero vectors.

We have $C_i^\dag C_j=\braket{v_i}{v_j}\ketbra{\phi_i}{\phi_j}$.
Hence
$$ \r_B=\sum_i C_i^\dag C_i=\sum_i \|v_i\|^2
\left(\begin{array}{cc} |\a_i|^2 & \a_i\b_i^* \\
\a_i^*\b_i & |\b_i|^2 \end{array}\right). $$
By using the Lagrange identity, we find that
\begin{eqnarray*}
\det(\r_B) &=&
\left( \sum_i |\a_i|^2\|v_i\|^2 \right)
\left( \sum_i |\b_i|^2\|v_i\|^2 \right)
-\left| \sum_i |\a_i\b_i^*|\|v_i\|^2 \right|^2 \\
&=& \sum_{i<j} \|v_i\|^2\|v_j\|^2|\a_i\b_j-\a_j\b_i|^2.
\end{eqnarray*}
Since all $\ket{v_i}\ne0$ and $\det(\r_B)>0$, there exists
a pair of indexes $i,j$ such that $\a_i\b_j-\a_j\b_i\ne0$.
Without any loss of generality we may assume that $i=1$ and $j=2$.
We know that the $R\times2$ matrix
\begin{eqnarray*}
\xi_1 C_1+\xi_2 C_2 &=&
\xi_1\ketbra{v_1}{\phi_1}+\xi_2\ketbra{v_2}{\phi_2} \\
&=& \left(
\xi_1\a_1^*\ket{v_1}+\xi_2\a_2^*\ket{v_2},
\xi_1\b_1^*\ket{v_1}+\xi_2\b_2^*\ket{v_2} \right)
\end{eqnarray*}
has rank less than two for arbitrary complex numbers $\xi_1$
and $\xi_2$. Thus all of its $2\times2$ minors must vanish:
\begin{eqnarray*}
0 &=&
\left| \begin{array}{cc}
\xi_1\a_1^*v_{1,k}+\xi_2\a_2^*v_{2,k} &
\xi_1\b_1^*v_{1,k}+\xi_2\b_2^*v_{2,k} \\
\xi_1\a_1^*v_{1,l}+\xi_2\a_2^*v_{2,l} &
\xi_1\b_1^*v_{1,l}+\xi_2\b_2^*v_{2,l} \end{array} \right| \\
&=&
\xi_1\xi_2(\a_1\b_2-\a_2\b_1)^*(v_{1,k}v_{2,l}-v_{2,k}v_{1,l}),
\end{eqnarray*}
where $1\le k<l\le M$ and $v_{i,k}$ denotes the $k$th
component of $\ket{v_i}$. Since $\xi_1$ and $\xi_2$ are arbitrary
and $\a_1\b_2-\a_2\b_1\ne0$, we have
$v_{1,k}v_{2,l}-v_{2,k}v_{1,l}=0$ for all $k<l$. Hence the vectors
$\ket{v_1}$ and $\ket{v_2}$ are linearly dependent.
Without any loss of generality we may assume that
$\ket{v_1}=\ket{v_2}$.

Recall from Eqs. (\ref{jed:ro-psi}) and (\ref{jed:mat-Cj}) that
we can write
$$ \r=\sum_{i=1}^R \proj{\psi_i}, $$
where $\ket{\psi_i}=\sum_j \ket{j}\ox \ket{\psi_{ij}}$ and
$\bra{\psi_{ij}}$ is the $i$th row of $C_j$, i.e.,
$$ \ket{\psi_{ij}}=v_{j,i}^*(\a_j\ket{1}+\b_j\ket{2})
=v_{j,i}^* \ket{\phi_j}. $$
Since $\ket{v_1}=\ket{v_2}$ we obtain that
$$ \ket{\psi_i}=v_{1,i}^*(\ket{1,\phi_1}+\ket{2,\phi_2})+
\sum_{j>2} v_{j,i}^*\ket{j,\phi_j}, \quad i=1,\ldots,R. $$
We conclude that $\cR(\r)$ is contained in the
subspace spanned by the vectors
$\ket{1,\phi_1}+\ket{2,\phi_2}$ and $\ket{j,\phi_j}$ for
$j=3,\ldots,R$. This contradicts the hypothesis that
$\rank(\r)=R$, and completes the proof.
\epf

In particular, all $2\times2$ states have LFRP and RFRP.
While all $2\times N$ states have LFRP, even the simplest
non-trivial case, i.e., a $2\times3$ state, may violate RFRP.
Indeed, one can easily verify that the $2\times3$ state
$\r=(\ket{11}+\ket{22})(\bra{11}+\bra{22})+\proj{23}+\proj{13}$
violates RFRP.

Hence for any $2\times N$ state $\r$ of rank $R\le N$ there exists a
state $\ket{x}\in\cH_A$ such that $\rank \bra{x} \r \ket{x} <N$.
However this is not true for $R>N$.

Finally, we have seen that the antisymmetric state $\r_{as}$ given
by Eq. (\ref{jed:ro-as}) violates RFRP. However, as one can easily
verify, its two-copy version $\r_{as}\ox\r_{as}$ has RFRP.

\section{\label{sec:conclusion} Conclusions}

Following \cite{hlv00}, we have introduced two full-rank properties,
LFRP and RFRP, and proved that the bipartite state violating at
least one of them is 1-distillable. We refer to this result as the
full-rank criterion for distillability. By using it, we obtained our first main result which asserts that the $M\times N$ NPT states of
rank $N$ are 1-distillable. In particular, this result provides the
affirmative solution to an open problem first proposed in 1999. This  result leads to a new characterization of the distillable entanglement,
namely a tripartite pure state $\ket{\ps}$ cannot have two
undistillable entangled bipartite reduced states. We also derived an
explicit expression for tripartite pure states having two
undistillable bipartite reduced states. Both of these states turn
out to be separable. On the other hand, we define reducible and
irreducible states and use them to distill some entangled states
possessing both LFRP and RFRP. The most important result in this
direction is that the checkerboard states \cite{dzd11} are
distillable under LOCC if and only if they are NPT.

We now list the most interesting results proved in this paper.
We shall just continue the enumeration of the results
(A-E) from the Introduction.
 \bem
\item[(F)]
Bipartite $M\times N$ NPT states of rank $\max(M,N)$ are
1-distillable.
\item[(G)]
A bipartite rank-4 state is separable if and only if it is PPT
and its range contains at least one product state.
\item[(H)]
If a bipartite state violates at least one of the two full-rank
conditions, then it is 1-distillable.
\item[(I)]
For a tripartite pure state $\ket{\psi}$, all three bipartite
reduced states are undistillable if and only if $\ket{\psi}$
is a generalized GHZ state.
\item[(J)]
The NPT checkerboard states (acting on a $3\ox 3$ system) are
1-distillable.
 \eem

For convenience, let us summarize what is known and what remains to
be done to finish off the coarse classification of rank 4 states
$\r$ in a $3\ox 3$ system. There are two cases:

1) $\cR(\r)$ contains a product state;

2) $\cR(\r)$ contains no product state.

In case 1), we know that if $\r$ is PPT then it is separable, and
otherwise it is 1-distillable.

In case 2), $\r$ is always entangled and may be PPT or NPT. For the
PPT subcase, there is a very precise conjecture \cite{lms10}
describing how these PPT entangled states can be generated from
UPBs. For the NPT subcase, we conjecture that these states are
distillable, and we know that this is true for the checkerboard
states (see Theorem~\ref{thm:checkboard}).

There are many related open problems and conjectures for further
study that originate from this paper. First, we have seen from
Proposition \ref{pp:irreducibledistillable,commonkernel} that an
irreducible bipartite state $\r$, for which there exists a nonzero
vector $\ket{x}$ such that $C_i\ket{x}=0$ for all but one of the
indexes $i$, is distillable. (Note that if $C_i\ket{x}=0$ then
$\ket{i,x}\in\ker\r$.) It is thus natural to ask whether
we can relax the above condition on the kernels of the $C_i$s.
For example one can conjecture that any irreducible $M\times N$ NPT state $\r$, with one of the $C_i$ of deficient rank, is distillable?  Since Werner states have full rank, all of its blocks $C_i$ also
have full rank, and so they do not contradict the conjecture.
Because of the rank condition on the $C_i$, the question evidently depends on the analysis of product states in the kernel of $\r$.
For example, it is a well-known fact that in an $M\ox  N$
system every subspace $V\subseteq\cH_A\ox\cH_B$ with
$\dim V>(M-1)(N-1)$ contains a product vector. To prove this,
we can just apply \cite[Proposition 11.4]{h92} to the Segre variety
consisting of all product vectors.

If the above conjecture turns out to be true, it would follow that
every bipartite $M\times N$ $(M\le N)$ NPT state $\r$ of rank 4 is
distillable. Indeed, if $M<3$ then $\r$ is distillable by result (A).
If $N>4$ then $\r$ is distillable by result (C). Thus we may assume
that $M=N=3$. As $\ker\r$ has dimension 5, it contains a product
state. By the conjecture, $\r$ is distillable.

To state the second conjecture, let $\r$ be a bipartite
NPT state acting on some $M\ox N$ system. We can view
$\r^{\ox k}$ as a bipartite state acting on a $M^k\ox N^k$ system.
Then the conjecture claims that for some $k\ge1$, the state
$\r^{\ox k}$ can be locally transformed into an $m\times n$ NPT
state of rank $\max(m,n)+1$.

If this conjecture is true, then we can project some tensor power of
an entangled 1-undistillable Werner state $\r_w$ onto an $m\times n$
$(m\le n)$ NPT state $\s$ of rank $n+1$. Because it is widely
believed that $\r_w$ is undistillable, we can thus conjecture that
there exists an undistillable $M\times N$ NPT state of rank $N+1$.
This would then imply that the value $N$ for the rank of the states
in Theorem~\ref{thm:distillation=MxNrankN} is maximal.

The third question is about the irreducible states. From
Lemma~\ref{le:reducible=SUMirreducible} we know that the
distillation problem relies on further investigation of irreducible
states. For example we ask: can the tensor product of
two irreducible states be a reducible state?
The existence of such phenomenon would become a sort of activation
of reducibility, which is akin to the activation of PPT bound entanglement \cite{hhh99}.

We also have seen that all $M\times N$ NPT states of rank $N$ are
1-distillable. In fact previously researchers have shown that such
states cannot be PPT entangled. However it is also well-known that
there are PPT entangled $M\times N$ states with rank bigger than
$N$, and meanwhile, there exist 1-undistillable but $n$-distillable
NPT states such as the Watrous state \cite{watrous04}. Is this an
essential difference between the above two families of states? Does
the existence of PPT entanglement imply that of 1-undistillable NPT
states?

Another interesting problem is to decide which PPT states are
entangled. For example, what is the relationship between the existence of product states in the range and entanglement of
PPT states? Can we conjecture that if there are more known
product states in the range of a PPT state, then it becomes easier
to decide whether it is entangled?

\acknowledgments

We thank two anonymous referees for their valuable suggestions.
Consequently the paper has been thoroughly reorganized and
considerably improved. We thank Andreas Winter for pointing out that
the $3\times3$ antisymmetric state violates RFRP. We also thank  
L. Skowronek for his e-mail \cite{skow} which led to a minor correction in the original version of the paper. Part of this work
was done when LC was visiting Prof. Heng Fan at the Institute of
Physics, CAS, Beijing, China and Prof. Runyao Duan at the University
of Technology, Sydney, Australia. The CQT is funded by the Singapore
MoE and the NRF as part of the Research Centres of Excellence
programme. The second author was supported in part by an NSERC
Discovery Grant.

\section{\label{sec:dodatak} Appendix}

We give here the equation which is necessary and sufficient for
a 3-dimensional subspace $V$ of a $2\ox 4$ system to contain at least
one product state. Let $\{\ket{a},\ket{b},\ket{c}\}$ be an arbitrary
basis of $V$. Thus
$$ \ket{a}=\sum_{i=1}^2 \sum_{j=1}^4 a_{ij}\ket{i,j}, $$
and similarly for $\ket{b}$ and $\ket{c}$.
Let us form the $3\times8$ matrix
from the components of these three states:
$$
P=\left( \begin{array}{cccccccc}
a_{11}& a_{12}& a_{13}& a_{14} &a_{21} &a_{22} &a_{23} &a_{24} \\
b_{11}& b_{12}& b_{13}& b_{14} &b_{21} &b_{22} &b_{23} &b_{24} \\
c_{11}& c_{12}& c_{13}& c_{14} &c_{21} &c_{22} &c_{23} &c_{24} \\
\end{array} \right). $$
Denote by $P^{ijk}$, $1\le i<j<k\le8$, the $3\times3$ submatrix
of $P$ made up from the $i$th, $j$th and $k$th columns.
The Pl{\"u}cker coordinates of $V$ are the 56 determinants
$p_{ijk}=\det P^{ijk}$. If we change the basis, the Pl{\"u}cker
coordinates of $V$ will be changed only by an overall factor.
The equation we aluded to is the following:
$$ p_{123}F_1+p_{124}F_2+p_{134}F_3+p_{234}F_4=0, $$
where
\begin{eqnarray*}
F_1 &=&
3p_{124}p_{578}p_{678}-3p_{125}p_{478}p_{678}+p_{126}p_{478}p_{578}
+2p_{127}p_{458}p_{678}-p_{127}p_{468}p_{578}+p_{128}p_{178}p_{678}
-p_{128}p_{278}p_{578} \\
&& -2p_{128}p_{358}p_{678}+p_{128}p_{368}p_{578}-3p_{134}p_{568}p_{678}
+3p_{135}p_{468}p_{678}-p_{136}p_{458}p_{678}-p_{136}p_{468}p_{578}
+p_{137}p_{468}p_{568} \\
&& -p_{138}p_{168}p_{678}+p_{138}p_{258}p_{678}+p_{138}p_{268}p_{578}
-p_{138}p_{368}p_{568}+2p_{145}p_{278}p_{678}-p_{145}p_{368}p_{678}
-2p_{145}p_{467}p_{678} \\
&&
-2p_{146}p_{178}p_{678}+2p_{146}p_{278}p_{578}+4p_{146}p_{358}p_{678}
-2p_{146}p_{368}p_{578}-2p_{146}p_{457}p_{678}+2p_{146}p_{467}p_{578}
+p_{147}p_{168}p_{678}\\
&& -p_{147}p_{258}p_{678}-p_{147}p_{268}p_{578}+p_{147}p_{368}p_{568}
-p_{147}p_{467}p_{568}+3p_{234}p_{568}p_{578}-2p_{235}p_{458}p_{678}
-2p_{235}p_{468}p_{578} \\
&& +2p_{236}p_{458}p_{578}-p_{237}p_{458}p_{568}-p_{238}p_{258}p_{578}
+p_{238}p_{358}p_{568}-2p_{245}p_{278}p_{578}-2p_{245}p_{358}p_{678}
+4p_{245}p_{368}p_{578} \\
&& +3p_{245}p_{457}p_{678} -p_{245}p_{467}p_{578}
-p_{246}p_{358}p_{578}-p_{246}p_{457}p_{578}+p_{247}p_{258}p_{578}
-p_{247}p_{358}p_{568}+p_{247}p_{457}p_{568} \\
&& -2p_{345}p_{368}p_{568}-3p_{345}p_{456}p_{678}+2p_{345}p_{467}p_{568}+p_{346}p_{358}p_{568}+p_{346}p_{456}p_{578}
-p_{346}p_{457}p_{568},
\end{eqnarray*}
\begin{eqnarray*}
F_2 &=&
p_{125}p_{378}p_{678}-p_{127}p_{178}p_{678}+p_{127}p_{278}p_{578}
+2p_{127}p_{358}p_{678}-p_{127}p_{368}p_{578}-2p_{127}p_{457}p_{678}
+p_{127}p_{467}p_{578} \\
&& +3p_{134}p_{567}p_{678}-p_{135}p_{278}p_{678}-p_{135}p_{368}p_{678}
-2p_{135}p_{467}p_{678}+p_{136}p_{178}p_{678}-p_{136}p_{278}p_{578}
-2p_{136}p_{358}p_{678} \\
&& +p_{136}p_{368}p_{578}+4p_{136}p_{457}p_{678}-p_{136}p_{467}p_{578}
+p_{137}p_{168}p_{678}-p_{137}p_{258}p_{678}-p_{137}p_{268}p_{578}
+p_{137}p_{368}p_{568} \\
&& -p_{137}p_{467}p_{568}+2p_{145}p_{367}p_{678}-p_{146}p_{357}p_{678}
-p_{147}p_{167}p_{678}+p_{147}p_{257}p_{678}+p_{147}p_{267}p_{578}
-p_{147}p_{367}p_{568} \\
&& +p_{147}p_{467}p_{567}-3p_{234}p_{567}p_{578}+p_{235}p_{278}p_{578}
+2p_{235}p_{358}p_{678}-p_{235}p_{368}p_{578}-p_{235}p_{457}p_{678}
+4p_{235}p_{467}p_{578} \\
&& -2p_{236}p_{457}p_{578}+p_{237}p_{258}p_{578}-p_{237}p_{358}p_{568}
+p_{237}p_{457}p_{568}-p_{245}p_{357}p_{678}-2p_{245}p_{367}p_{578}
+p_{246}p_{357}p_{578} \\
&& -p_{247}p_{257}p_{578}+p_{247}p_{357}p_{568}-p_{247}p_{457}p_{567}
+p_{345}p_{356}p_{678}+2p_{345}p_{367}p_{568}-2p_{345}p_{467}p_{567}
-p_{346}p_{357}p_{568} \\
&& +p_{346}p_{457}p_{567},
\end{eqnarray*}
\begin{eqnarray*}
F_3 &=&
p_{135}p_{268}p_{678}-p_{136}p_{168}p_{678}+p_{136}p_{258}p_{678}
+p_{136}p_{268}p_{578}-p_{136}p_{368}p_{568}-2p_{136}p_{456}p_{678}
+p_{136}p_{467}p_{568} \\
&& -p_{145}p_{267}p_{678}+p_{146}p_{167}p_{678}-p_{146}p_{257}p_{678}
-p_{146}p_{267}p_{578}+p_{146}p_{356}p_{678}+p_{146}p_{367}p_{568}
-p_{146}p_{467}p_{567} \\
&& -p_{235}p_{258}p_{678}-p_{235}p_{268}p_{578}+p_{235}p_{368}p_{568}
+2p_{235}p_{456}p_{678}-p_{235}p_{467}p_{568}-p_{236}p_{258}p_{578}
+p_{236}p_{358}p_{568} \\
&& +2p_{236}p_{456}p_{578}-p_{236}p_{457}p_{568}+p_{245}p_{257}p_{678}
+p_{245}p_{267}p_{578}-p_{245}p_{356}p_{678}-p_{245}p_{367}p_{568}
+p_{245}p_{467}p_{567} \\
&& +p_{246}p_{257}p_{578}-p_{246}p_{356}p_{578}-p_{246}p_{357}p_{568}
+p_{246}p_{457}p_{567}+p_{346}p_{356}p_{568}-p_{346}p_{456}p_{567},
\end{eqnarray*}
\begin{eqnarray*}
F_4 &=&
p_{235}p_{258}p_{578}-p_{235}p_{358}p_{568}-2p_{235}p_{456}p_{578}
+p_{235}p_{457}p_{568}-p_{245}p_{257}p_{578}+p_{245}p_{356}p_{578}
+p_{245}p_{357}p_{568} \\
&& -p_{245}p_{457}p_{567}-p_{345}p_{356}p_{568}+p_{345}p_{456}p_{567}.
\end{eqnarray*}
Although this equation looks complicated, it can be easily
programmed on a computer and used to decide whether $V$ contains
a product state. The main point is that we do not need to solve
numerically any algebraic equations.

\end{document}